\documentclass[conference,compsoc]{IEEEtran}
\usepackage{fix-cm} %

\newif\ifanonymous
\anonymousfalse
\newif\ifreview
\reviewfalse
\newif\ifacm
\acmfalse
\newif\ifgeolife
\geolifetrue

\usepackage{mystyle} %
\newacronym{dow}{DoW}{Day of Week}

\newacronym{bce}{BCE}{Binary Cross Entropy}

\newacronym{ce}{CE}{Cross Entropy}

\newacronym{dp-stg}{DP-STG}{Differentially Private Synthetic Trajectory Generator}

\newacronym{edr}{EDR}{Edit Distance on Real Sequences}

\newacronym{exGAN}{exGAN}{Except-Condition GAN}

\newacronym{gp}{GP}{General Practitioner}

\newacronym{jsd}{JSD}{Jensen Shannon Divergence}

\newacronym{kde}{KDE}{Kernel Density Estimation}

\newabbreviation{TM}{TM}{Threat Model}

\newacronym{llm}{LLM}{Large Language Model}

\newabbreviation{wrt}{w.r.t.}{with respect to}

\newglossaryentry{mae}{
    name={MAE},
    short={MAE},
    long={Mean Absolute Error},
    first={Mean Absolute Error (MAE)},
    description={Mean Absolute Error}
}

\newglossaryentry{mre}{
    name={MRE},
    short={MRE},
    long={Mean Relative Error},
    first={Mean Relative Error (MRE)},
    description={Mean Relative Error}
}

\newacronym{mi}{MI}{Mutual Information}

\newacronym{mia}{MIA}{Membership Interference Attack}

\newacronym{ml}{ML}{Machine Learning}

\newacronym{mlp}{MLP}{Multi-Layer Perceptron}

\newacronym{mse}{MSE}{Mean Squared Error}

\newacronym{mwe}{MWE}{Minimal Working Example}

\newacronym{sgd}{SGD}{Stochastic Gradient Descent}

\newacronym{sota}{SOTA}{State Of The Art}

\newacronym{stc}{STC}{Spatial-Temporal-Categorical Distance}

\newacronym{stg}{STG}{Synthetic Trajectory Generator}

\newacronym{stn}{STN}{Large Spatial Transformer Network}

\newacronym{tcac}{TCAC}{Trajectory Category Auxiliary Classifier}

\newacronym{TRA}{TRA}{Trajectory Reconstruction Attack}

\newacronym{TSG}{TSG}{Two-Stage-\glsshort{gan}}

\newglossaryentry{aae}{
    name={AAE},
    short={AAE},
    long={Adversarial Autoencoder},
    first={Adversarial Autoencoder (AAE)},
    shortplural={AAEs},
    longplural={Adversarial Autoencoders},
    firstplural={Adversarial Autoencoders (AAEs)},
    description={Adversarial Autoencoder}
}

\newglossaryentry{adam}{
    name={Adam},
    short={Adam},
    description={Adam Optimizer~\cite{adam-optimizer}}
}

\newglossaryentry{adamw}{
    name={AdamW},
    short={AdamW},
    long={AdamW},
    first={AdamW},
    description={Pytorch's Adamw Optimizer with Decoupled Weight Decay}
}

\newglossaryentry{ae}{
    name={AE},
    short={AE},
    shortplural={AEs},
    long={Autoencoder},
    first={Autoencoder (AE)},
    firstplural={Autoencoders (AEs)},
    description={Autoencoder}
}

\newglossaryentry{ar}{
    name={AR-RNN},
    short={AR-RNN},
    long={Autoregressive RNN},
    first={Autoregressive \glsentryshort{rnn} (AR-RNN)},
    description={Autoregressive \glsentryshort{rnn}}
}

\newglossaryentry{cnn}{
    name={CNN},
    short={CNN},
    shortplural={CNNs},
    long={Convolutional Neural Network},
    first={Convolutional Neural Network (CNN)},
    firstplural={Convolutional Neural Networks (CNNs)},
    description={Convolutional Neural Network}
}

\newglossaryentry{conv1d}{
    name={Conv1D},
    short={Conv1D},
    shortplural={Conv1Ds},
    description={1D Convolution},
    first={Conv1D},
}

\newglossaryentry{dcgan}{
    name={DCGAN},
    short={DCGAN},
    shortplural={DCGANs},
    long={Deep Convolutional Generative Adversarial Network},
    first={Deep Convolutional Generative Adversarial Network (DCGAN)},
    firstplural={Deep Convolutional Generative Adversarial Networks (DCGANs)},
    description={Deep Convolutional Generative Adversarial Network}
}

\newglossaryentry{ddpm}{
    name={DDPM},
    short={DDPM},
    shortplural={DDPMs},
    long={Denoising Diffusion Probabilistic Model},
    first={Denoising Diffusion Probabilistic Model (DDPM)},
    firstplural={Denoising Diffusion Probabilistic Models (DDPMs)},
    description={Denoising Diffusion Probabilistic Model}
}

\newglossaryentry{dl}{
    name={DL},
    short={DL},
    long={Deep Learning},
    first={Deep Learning (DL)},
    description={Deep Learning}
}

\newglossaryentry{dp-sgd}{
    name={DP-SGD},
    short={DP-SGD},
    first={Differentially Private Stochastic Gradient Descent (DP-SGD)},
    long={Differentially Private Stochastic Gradient Descent},
    description={Differentially Private Stochastic Gradient Descent}
}

\newglossaryentry{dp}{
    name={DP},
    short={DP},
    long={Differential Privacy},
    first={Differential Privacy (DP)},
    description={Differential Privacy}
}

\newglossaryentry{dtw}{
    name={DTW},
    short={DTW},
    long={Dynamic Time Warping},
    first={Dynamic Time Warping (DTW)},
    description={Dynamic Time Warping}
}

\newglossaryentry{eid}{
    name={Explicit Identifier},
    first={Explicit Identifier},
    firstplural={Explicit Identifiers},
    description={Attributes that uniquely identify an individual},
}

\newglossaryentry{em}{
    name={EM},
    short={EM},
    shortplural={EMs},
    long={Exponential Mechanism},
    first={Exponential Mechanism (EM)},
    description={Exponential Mechanism}
}

\newglossaryentry{ema}{
    name={EMA},
    short={EMA},
    long={Exponential Moving Average},
    first={Exponential Moving Average (EMA)},
    description={Exponential Moving Average}
}

\newglossaryentry{emd}{
    name={EMD},
    short={EMD},
    long={Earth Mover's Distance},
    first={Earth Mover's Distance (EMD)},
    description={Earth Mover's Distance}
}

\newglossaryentry{fc}{
    name={FC},
    short={FC},
    shortplural={FCs},
    description={Fully Connected layer. Also called \textit{Dense} (TensorFlow) or \textit{Linear} (PyTorch) layer},
    first={Fully Connected (FC)},
    long={Fully Connected}
}

\newglossaryentry{fs}{
    name={FS-NYC},
    short={FS-NYC},
    long={Foursquare NYC},
    first={Foursquare NYC (FS-NYC)},
    description={Foursquare NYC~\cite{fs_nyc}}
}

\newglossaryentry{gan}{
    name={GAN},
    short={GAN},
    shortplural={GANs},
    long={Generative Adversarial Network},
    first={Generative Adversarial Network (GAN)},
    firstplural={Generative Adversarial Networks (GANs)},
    description={Generative Adversarial Network}
}

\newglossaryentry{geo-ind}{
    name={Geo-Ind},
    short={Geo-Ind},
    long={Geo-Indistinguishability},
    first={Geo-Indistinguishability (Geo-Ind)},
    description={Geo-Indistinguishability~\cite{Andres2013}}
}

\newglossaryentry{gpg}{
    name={GPG},
    short={GPG},
    first={GeoPointGAN (GPG)},
    description={GeoPointGAN~\cite{GeoPointGAN}}
}

\newglossaryentry{gru}{
    name={GRU},
    short={GRU},
    shortplural={GRUs},
    long={Gated Recurrent Unit},
    first={Gated Recurrent Unit (GRU)},
    firstplural={Gated Recurrent Units (GRUs)},
    description={Gated Recurrent Unit}
}

\newglossaryentry{gtg}{
    name={GTG},
    short={GTG},
    long={GeoTrajGAN},
    first={GeoTrajGAN (GTG)},
    description={GeoTrajGAN}
}

\newglossaryentry{hd}{
    name={HD},
    short={HD},
    long={Hausdorff Distance},
    first={Hausdorff Distance (HD)},
    firstplural={Hausdorff Distances (HDs)},
    description={Hausdorff Distance}
}

\newglossaryentry{iid}{
    name={i.i.d.},
    description={independent and identically distributed random variables}
}

\newglossaryentry{ladp}{
    name={LaDP},
    short={LaDP},
    long={Label Differential Privacy},
    first={Label Differential Privacy (LaDP)},
    description={Label Differential Privacy}
}

\newglossaryentry{ldp}{
    name={LDP},
    short={LDP},
    long={Local Differential Privacy},
    first={Local Differential Privacy (LDP)},
    description={Local Differential Privacy}
}

\newglossaryentry{lldp}{
    name={LLDP},
    short={LLDP},
    long={Local Label Differential Privacy},
    first={Local Label Differential Privacy (LLDP)},
    description={Local Label Differential Privacy}
}

\newglossaryentry{lstm}{
    name={LSTM},
    short={LSTM},
    shortplural={LSTMs},
    long={Long Short-Term Memory},
    first={Long Short-Term Memory (LSTM)},
    firstplural={Long Short-Term Memories (LSTMs)},
    description={Long Short-Term Memory}
}

\newglossaryentry{ltg}{
    name={LSTM-TrajGAN},
    short={LSTM-TrajGAN}
    first={LSTM-TrajGAN},
    description={LSTM-TrajGAN~\cite{Rao2020}}
}

\newglossaryentry{mnist-seq}{
    name={MNIST-Seq},
    first={MNIST Sequential (MNIST-Seq)},
    description={MNIST Sequential Dataset: Images are transformed to sequences of length $28$ with $28$ features each~\cite{Esteban2017}}
}

\newglossaryentry{nlp}{
    name={NLP},
    short={NLP},
    long={Natural Language Processing},
    first={Natural Language Processing (NLP)},
    description={Natural Language Processing}
}

\newglossaryentry{non-sensitive}{
    name={Non-Sensitive Attribute},
    long={non-sensitive attribute},
    longplural={non-sensitive attributes},
    first={non-sensitive attribute},
    firstplural={non-sensitive attributes},
    description={Attributes that do not contain sensitive information and are generally not used for individual identification},
}

\newglossaryentry{noise-only}{
    name={Noise-only},
    short={noise-only},
    long={noise-only},
    first={noise-only},
    description={A noise-only generative model is a model that only receives random noise, usually $z \sim \mathcal{N}(0,1)$, as input and generates samples from the data distribution}
}

\newglossaryentry{ntg}{
    name={NTG},
    short={NTG},
    long={Noise-TrajGAN},
    first={Noise-TrajGAN (NTG)},
    description={Noise-TrajGAN}
}

\newglossaryentry{nwa}{
    name={NWA},
    short={NWA},
    long={Newer Walk Alone},
    first={Newer Walk Alone (NWA)},
    description={Newer Walk Alone~\cite{Abul2008}}
}

\newglossaryentry{OP-Distance}{
    name=OP-Distance,
    description={Distance between  \emph{original} and \emph{protected} trajectory}
}

\newglossaryentry{OR-Distance}{
    name=OR-Distance,
	description={Distance between \emph{original} and \emph{reconstructed} trajectory}
}

\newglossaryentry{pdf}{
    name={PDF},
    short={PDF},
    long={Probability Density Function},
    first={Probability Density Function (PDF)},
    firstplural={Probability Density Functions (PDFs)},
    description={Probability Density Function}
}

\newglossaryentry{poc}{
    name={PoC},
    short={PoC},
    long={Proof-of-Concept},
    first={Proof-of-Concept (PoC)},
    firstplural={Proof-of-Concepts (PoCs)},
    description={Proof-of-Concept}
}

\newglossaryentry{poi}{
    name={POI},
    short={POI},
    shortplural={POIs},
    long={Point of Interest},
    first={Point of Interest (POI)},
    firstplural={Points of Interest (POIs)},
    description={Point of Interest}
}

\newglossaryentry{ppdp}{
    name={PPDP},
    short={PPDP},
    long={Privacy-Preserving Data Publication},
    first={Privacy-Preserving Data Publication (PPDP)},
    firstplural={Privacy-Preserving Data Publications (PPDPs)},
    description={Privacy-Preserving Data Publication}
}

\newglossaryentry{pptp}{
    name={PPTP},
    short={PPTP},
    long={Privacy-Preserving Trajectory Publication},
    first={Privacy-Preserving Trajectory Publication (PPTP)},
    firstplural={Privacy-Preserving Trajectory Publications (PPTPs)},
    description={Privacy-Preserving Trajectory Publication}
}

\newglossaryentry{qid}{
    name={Quasi-Identifier},
    long={quasi-identifier},
    longplural={quasi-identifiers},
    short={QID},
    shortplural={QIDs},
    description={Attributes that do not uniquely identify individuals alone but in combination with other quasi-identifiers},
    firstplural={Quasi-identifiers (QIDs)},
    first={Quasi-identifier (QID)}
}

\newglossaryentry{RAoPT}{
    name={RAoPT},
    short={RAoPT},
    long={Reconstruction Attack on Protected Trajectories},
    first={Reconstruction Attack on Protected Trajectories (RAoPT)},
    description={Reconstruction Attack on Protected Trajectories}
}

\newglossaryentry{rappor}{
    name={RAPPOR},
    short={RAPPOR},
    first={RAPPOR},
    description={Randomized Aggregatable Privacy-Preserving Ordinal Response~\cite{Erlingsson2014}}
}

\newglossaryentry{rdp}{
    name={RDP},
    short={RDP},
    long={R\'enyi Differential Privacy},
    first={R\'enyi Differential Privacy (RDP)},
    description={R\'enyi Differential Privacy}
}

\newglossaryentry{re}{
    name={RE},
    short={RE},
    long={Relative Error},
    first={Relative Error (RE)},
    description={Relative Error}
}

\newglossaryentry{relu}{
    name={ReLU},
    short={ReLU},
    shortplural={ReLUs},
    long={Rectified Linear Unit},
    longplural={Rectified Linear Units},
    first={Rectified Linear Unit (ReLU)},
    firstplural={Rectified Linear Units (ReLUs)},
    description={Rectified Linear Unit. Activation function used in neural networks},
}

\newglossaryentry{rgan}{
    name={RGAN},
    short={RGAN},
    long={Recurrent GAN},
    first={Recurrent \glsentryshort{gan} (RGAN)},
    description={Recurrent \glsentryshort{gan}}
}

\newglossaryentry{rnn}{
    name={RNN},
    short={RNN},
    shortplural={RNNs},
    long={Recurrent Neural Network},
    first={Recurrent Neural Network (RNN)},
    firstplural={Recurrent Neural Networks (RNNs)},
    description={Recurrent Neural Network}
}

\newglossaryentry{rtct}{
    name={RTCT},
    short={RTCT},
    long={Reversible Trajectory-to-CNN Transformation},
    first={Reversible Trajectory-to-CNN Transformation (RTCT)},
    description={Reversible Trajectory-to-CNN Transformation}
}

\newglossaryentry{sdd}{
    name={SDD},
    short={SDD},
    long={Sampling Distance and Direction},
    first={Sampling Distance and Direction (SDD)},
    description={Sampling Distance and Direction}
}

\newglossaryentry{sensitive}{
    name={Sensitive Attribute},
    long={sensitive attribute},
    short={SA},
    shortplural={SAs},
    longplural={sensitive attributes},
    first={Sensitive Attribute (SA)},
    firstplural={Sensitive Attributes (SAs)},
    description={Attributes containing sensitive information that should be protected, such as medical records or financial details},
}

\newglossaryentry{start}{
    name={START-RNN},
    short={START-RNN},
    long={Start-Point RNN},
    first={Start-Point \glsentryshort{rnn} (START-RNN)},
    description={Start-Point \glsentryshort{rnn}}
}

\newglossaryentry{swd}{
    name={SWD},
    short={SWD},
    long={Sliced Wasserstein Distance},
    first={Sliced Wasserstein Distance (SWD)},
    firstplural={Sliced Wasserstein Distances (SWDs)},
    description={Sliced Wasserstein Distance}
}

\newglossaryentry{tanh}{
    name=tanh,
    short=tanh,
    description=Hyperbolic Tangent
}

\newglossaryentry{trr}{
    name={TRR},
    short={TRR},
    long={Time Reversal Ratio},
    first={Time Reversal Ratio (TRR)},
    firstplural={Time Reversal Ratios (TRRs)},
    description={Time Reversal Ratio}
}

\newglossaryentry{ttd}{
    name={TTD},
    short={TTD},
    long={Total Travelled Distance},
    first={Total Travelled Distance (TTD)},
    firstplural={Total Travelled Distances (TTDs)},
    description={Total Travelled Distance}
}

\newglossaryentry{ttur}{
    name={TTUR},
    short={TTUR},
    long={Two Time Update Rule},
    first={Two Time Update Rule (TTUR)},
    firstplural={Two Time Update Rules (TTURs)},
    description={Two Time Update Rule}
}

\newglossaryentry{tul}{
    name={TUL},
    short={TUL},
    long={Trajectory User Linking},
    first={Trajectory User Linking (TUL)},
    description={Trajectory User Linking}
}

\newglossaryentry{uop}{
    name={UoP},
    short={UoP},
    long={Unit of Privacy},
    first={Unit of Privacy (UoP)},
    firstplural={Units of Privacy (UoPs)},
    description={Unit of Privacy}
}

\newglossaryentry{vae}{
    name={VAE},
    short={VAE},
    long={Variational Autoencoder},
    first={Variational Autoencoder (VAE)},
    firstplural={Variational Autoencoders (VAEs)},
    description={Variational Autoencoder}
}

\newglossaryentry{w4m}{
    name={W4M},
    short={W4M},
    long={Wait for Me},
    first={Wait for Me (W4M)},
    description={Wait for Me~\cite{Abul2010}}
}
\newglossaryentry{wd}{
    name={WD},
    short={WD},
    long={Wasserstein Distance},
    first={Wasserstein Distance (WD)},
    firstplural={Wasserstein Distances (WDs)},
    description={Wasserstein Distance}
}

\newglossaryentry{wgan-gp}{
    name={WGAN-GP},
    short={WGAN-GP},
    first={WGAN with Gradient Penalty (WGAN-GP)},
    firstplural={WGANs with Gradient Penalty (WGAN-GP)},
    description={WGAN with Gradient Penalty~\cite{iWGAN}},
    long={WGAN with Gradient Penalty}
}

\newglossaryentry{wgan-lp}{
    name={WGAN-LP},
    short={WGAN-LP},
    first={WGAN-LP},
    description={WGAN with Lipschitz Penalty~\cite{wgan_lp}},
    long={WGAN with Lipschitz Penalty}
}

\newglossaryentry{wgan}{
    name={WGAN},
    short={WGAN},
    shortplural={WGANs},
    first={Wasserstein GAN (WGAN)},
    firstplural={Wasserstein GANs (WGANs)},
    long={Wasserstein GAN},
    longplural={Wasserstein GANs},
    description={Wasserstein \glsentryshort{gan}\cite{wgan}}
}

\newglossaryentry{vmf}{
    name={VMF},
    short={VMF},
    long={Von Mises-Fisher Distribution},
    first={Von Mises-Fisher Distribution (VMF)},
    firstplural={Von Mises-Fisher Distributions (VMFs)},
    description={Von Mises-Fisher Distribution}
}

\newglossaryentry{hpc}{
    name={HPC},
    short={HPC},
    long={High-Performance Computing},
    first={High-Performance Computing (HPC)},
    description={High-Performance Computing}
}

\newglossaryentry{sdc}{
    name={SDC},
    short={SDC},
    long={S\o rensen-Dice coefficient},
    first={S\o rensen-Dice coefficient (SDC)},
    description={S\o rensen-Dice coefficient}
}

\glsadd{adamw}

\newcommand{\DTcond}{8-Stat\xspace}  %
\newcommand{\Sample}{Sample\xspace}  %
\newcommand{\condInfo}{cond.\ info.\xspace}  %
\newcommand{\wrt}{\gls{wrt}\xspace}

\newglossaryentry{n}{
    name={\ensuremath{n}},
    description={Number of Samples in the Dataset},
    type=symbols
}

\newglossaryentry{m}{
    name={\ensuremath{m}},
    description={Number of Samples to be Generated with $m < n$},
    type=symbols
}

\newglossaryentry{L}{
    name={\ensuremath{L}},
    description={Trajectory Length},
    type=symbols
}

\newglossaryentry{p}{
    name={\ensuremath{p}},
    description={Type of Norm. $p=1$ for $l_1$-Norm, $p=2$ for $l_2$-Norm},
    type=symbols
}

\newglossaryentry{C}{
    name={\ensuremath{C}},
    description={Norm Bound},
    type=symbols
}

\newglossaryentry{M_c}{
    name={\ensuremath{M_c}},
    description={Private Compression Function},
    type=symbols
}

\newglossaryentry{f_c}{
    name={\ensuremath{f_c}},
    description={Sample Compression Function},
    type=symbols
}

\newglossaryentry{F_c}{
    name={\ensuremath{F_c}},
    description={Dataset Compression Function},
    type=symbols
}

\newglossaryentry{epsilon}{
    name={\ensuremath{\epsilon}},
    description={Privacy Parameter},
    type=symbols
}

\newglossaryentry{lap}{
    name={\ensuremath{\mathcal{L}(\lambda)}},
    text={\ensuremath{\mathcal{L}}},
    description={Laplace Distribution with scale $\lambda$},
    type=symbols
}

\newglossaryentry{gau}{
    name={\ensuremath{\mathcal{N}(\mu, \sigma^2)}},
    text={\ensuremath{\mathcal{N}}},
    description={Normal (Gaussian) Distribution with mean $\mu$ and variance $\sigma^2$},
    type=symbols
}

\newglossaryentry{uni}{
    name={\ensuremath{\mathcal{U}(a, b)}},
    text={\ensuremath{\mathcal{U}}},
    description={Uniform Distribution with bounds $a$ and $b$},
    type=symbols
}

\newglossaryentry{V}{
    name={\ensuremath{\mathcal{V}(\kappa, \mu)}},
    text={\ensuremath{\mathcal{V}}},
    description={Von Mises-Fisher Distribution with mean $\mu$ and concentration $\kappa$},
    type=symbols
}

\newglossaryentry{l1-norm}{
    name={\ensuremath{\| \cdot \|_1}},
    description={$l_1$-Norm (Manhattan Norm)},
    type=symbols
}
\glsadd{l1-norm}

\newglossaryentry{l2-norm}{
    name={\ensuremath{\| \cdot \|_2}},
    description={$l_2$-Norm (Euclidean Norm)},
    type=symbols
}
\glsadd{l2-norm}

\newglossaryentry{lp-norm}{
    name={\ensuremath{\| \cdot \|_p}},
    description={$l_p$-Norm},
    type=symbols
}
\glsadd{lp-norm}

\newglossaryentry{Delta_1}{
    name={\ensuremath{\Delta_1}},
    description={$l_1$-Sensitivity},
    type=symbols
}

\newglossaryentry{Delta_2}{
    name={\ensuremath{\Delta_2}},
    description={$l_2$-Sensitivity},
    type=symbols
}

\newglossaryentry{Delta_p}{
    name={\ensuremath{\Delta_p}},
    description={$l_p$-Sensitivity},
    type=symbols
}

\newglossaryentry{adj}{
    name={\ensuremath{Adj}},
    description={Adjacency},
    type=symbols
}

\newglossaryentry{prob}{
    name={\ensuremath{\mathds{P}}},
    description={Probability},
    type=symbols
}

\newglossaryentry{reels}{
    name={\ensuremath{\mathbb{R}}},
    description={Real Numbers},
    type=symbols
}

\newglossaryentry{x}{
    name={\ensuremath{x}},
    description={Input Data},
    type=symbols
}

\newglossaryentry{y}{
    name={\ensuremath{y}},
    description={Ground Truth},
    type=symbols
}

\newglossaryentry{x_i}{
    name={\ensuremath{x_i}},
    description={Input Data Sample},
    type=symbols
}

\newglossaryentry{y_i}{
    name={\ensuremath{y_i}},
    description={Ground Truth for Input Sample $x_i$},
    type=symbols
}

\newglossaryentry{yhat_i}{
    name={\ensuremath{\hat{y}_i}},
    description={Predicted Output for Input Sample $x_i$},
    type=symbols
}

\newglossaryentry{parameters}{
    name={\ensuremath{\theta}},
    description={Trainable Parameters of the Model},
    type=symbols
}

\newglossaryentry{mechanism}{
    name={\ensuremath{\mathcal{K}}},
    description={Privacy Mechanism},
    type=symbols
}

\newglossaryentry{input_space}{
    name={\ensuremath{\mathcal{X}}},
    description={Input Space},
    type=symbols
}

\newglossaryentry{output_space}{
    name={\ensuremath{\mathcal{Y}}},
    description={Output Space},
    type=symbols
}

\newglossaryentry{loss}{
    name={\ensuremath{\mathcal{L}}},
    description={Loss Function $\mathcal{L}(y_i, \hat{y}_i)$},
    type=symbols
}

\newglossaryentry{gradient}{
    name={\ensuremath{\nabla_\theta}},
    description={Gradient with Respect to Trainable Parameters $\theta$},
    type=symbols
}

\newglossaryentry{yhat}{
    name={\ensuremath{\hat{y}}},
    description={Predicted Output for Input $x$},
    type=symbols
}

\newglossaryentry{d_in}{
    name={\ensuremath{d_{in}}},
    description={Input Dimension},
    type=symbols
}

\newglossaryentry{d_out}{
    name={\ensuremath{d_{out}}},
    description={Output Dimension},
    type=symbols
}

\newglossaryentry{D}{
    name={\ensuremath{D}},
    description={Dataset},
    type=symbols
}

\newglossaryentry{D_real}{
    name={\ensuremath{D_{\text{real}}}},
    description={Real Dataset},
    type=symbols
}

\newglossaryentry{D_syn}{
    name={\ensuremath{D_{\text{syn}}}},
    description={Synthetic Dataset},
    type=symbols
}

\newglossaryentry{D_test}{
    name={\ensuremath{D_{\text{test}}}},
    description={Test Dataset. Dataset during Generation.},
    type=symbols
}

\newglossaryentry{D_train}{
    name={\ensuremath{D_{\text{train}}}},
    description={Training Dataset},
    type=symbols
}

\newglossaryentry{eps_s}{
    name={\ensuremath{\varepsilon_s}},
    description={Privacy loss parameter used in \gls{dp-sgd}},
    type=symbols
}

\newglossaryentry{delta_s}{
    name={\ensuremath{\delta_s}},
    description={Failure probability (imprecise but commonly used term~\cite{deltaInDp2020}) used in \gls{dp-sgd}},
    type=symbols
}

\newglossaryentry{eps_c}{
    name={\ensuremath{\varepsilon_c}},
    description={Privacy loss parameter used for conditional information},
    type=symbols
}

\newglossaryentry{delta_c}{
    name={\ensuremath{\delta_c}},
    description={Failure probability (imprecise but commonly used term~\cite{deltaInDp2020}) used for \condInfo},
    type=symbols
}

\newcommand{\mainRQ}{\textit{What is the (utility) cost of formal privacy guarantees for deep learning-based trajectory generation}}
\newcommand{\rqDPSGD}{\textit{How does \gls{dp-sgd} affect the generated trajectories' utility?}\xspace}
\newcommand{\rqDPCond}{\textit{How can we use conditional information with differential privacy guarantees?}\xspace}
\newcommand{\rqArchitecture}{\textit{Are certain model types better suited for formal privacy guarantees than others?}\xspace}

\newcommand{\tashsays}[1]{#1}
\newcommand{\update}[1]{#1}

\ifCLASSINFOpdf
\else
\fi

\ifCLASSOPTIONcompsoc
 \usepackage[caption=false,font=footnotesize,labelfont=sf,textfont=sf]{subfig}
\else
 \usepackage[caption=false,font=footnotesize]{subfig}
\fi

\begin{document}
\title{What is the Cost of Differential Privacy for Deep Learning-Based Trajectory Generation?}

\ifanonymous
  \author{\IEEEauthorblockN{Anonymous Authors}}
\else
\author{
\IEEEauthorblockN{Erik Buchholz}
\IEEEauthorblockA{University of New South Wales\\
CSIRO's Data61, Cyber Security CRC\\
Sydney, NSW, Australia\\
\href{mailto:e.buchholz@unsw.edu.au}{e.buchholz@unsw.edu.au}}
\and
\IEEEauthorblockN{Natasha Fernandes}
\IEEEauthorblockA{Macquarie University\\
Sydney, NSW, Australia\\
\href{mailto:natasha.fernandes@mq.edu.au}{natasha.fernandes@mq.edu.au}}
\and
\IEEEauthorblockN{David D. Nguyen}
\IEEEauthorblockA{CSIRO's Data61\\
Sydney, NSW, Australia\\
\href{mailto:david.nguyen@data61.csiro.au}{david.nguyen@data61.csiro.au}}
\and
\IEEEauthorblockN{Alsharif Abuadbba}
\IEEEauthorblockA{CSIRO's Data61, Cyber Security CRC\\
Sydney, NSW, Australia\\
\href{mailto:sharif.abuadbba@data61.csiro.au}{sharif.abuadbba@data61.csiro.au}}
\and
\IEEEauthorblockN{Surya Nepal}
\IEEEauthorblockA{CSIRO's Data61, Cyber Security CRC\\
Sydney, NSW, Australia\\
\href{mailto:surya.nepal@data61.csiro.au}{surya.nepal@data61.csiro.au}}
\and
\IEEEauthorblockN{Salil S. Kanhere}
\IEEEauthorblockA{University of New South Wales\\
Sydney, NSW, Australia\\
\href{mailto:salil.kanhere@unsw.edu.au}{salil.kanhere@unsw.edu.au}}
}
\fi

\maketitle

\begin{abstract}
    While location trajectories offer valuable insights, they also reveal sensitive personal information.
Differential Privacy (DP) offers formal protection, but achieving a favourable utility-privacy trade-off remains challenging.
Recent works explore deep learning-based generative models to produce synthetic trajectories.
However, current models lack formal privacy guarantees and rely on conditional information derived from real data during generation.
This work investigates the utility cost of enforcing DP in such models, addressing three research questions across two datasets and eleven utility metrics.
(1) We evaluate how DP-SGD, the standard DP training method for deep learning, affects the utility of state-of-the-art generative models.
(2) Since DP-SGD is limited to unconditional models, we propose a novel DP mechanism for conditional generation that provides formal guarantees and assess its impact on utility. 
(3) We analyse how model types -- Diffusion, VAE, and GAN -- affect the utility-privacy trade-off.
Our results show that DP-SGD significantly impacts performance, although some utility remains if the datasets is sufficiently large.
The proposed DP mechanism improves training stability, particularly when combined with DP-SGD, for unstable models such as GANs and on smaller datasets.
Diffusion models yield the best utility without guarantees, but with DP-SGD, GANs perform best, indicating that the best non-private model is not necessarily optimal when targeting formal guarantees.
In conclusion, DP trajectory generation remains a challenging task, and formal guarantees are currently only feasible with large datasets and in constrained use cases.

\end{abstract}

\IEEEpeerreviewmaketitle

\section{Introduction}\label{sec_intro}

Various services, such as navigation (Google Maps), transport (Uber), and social gaming (Pok\'emon Go), collect vast amounts of location data.
While the collected location trajectories are valuable for various applications, from urban planning to disease spread analysis, they also contain sensitive information about individuals~\cite{Primault2019}.
For example, four locations suffice to identify \SI{95}{\%} of users from phone connection data~\cite{DeMontjoye2013}, and taxi drivers' break times can reveal whether they are practising Muslims~\cite{Franceschi-Bicchierai}.
Accordingly, a large body of research is dedicated to the privacy of trajectory data~\cite{Primault2019,Jiang2021,Miranda-Pascual2024,pets24_paper}.
Due to their formal guarantees, significant research has focused on mechanisms with \gls{dp}~\cite{Miranda-Pascual2023}.
Yet, the challenging utility-privacy trade-off, vulnerability to reconstruction attacks~\cite{RAoPT,Shao2020}, and other limitations hinder their real-world applicability~\cite{Miranda-Pascual2023,pets24_paper}.

\begin{figure}[t]
  \centering
  \includegraphics[width=\columnwidth]{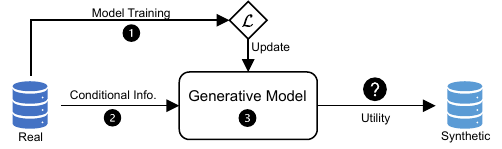}
  \vspace{-1.5em}
  \caption{
    Research questions:
    How do
    (1) \glsshort{dp-sgd},
    (2) \glsshort{dp} conditional information, and 
    (3) model architectures affect the utility and privacy trade-off of trajectory generation?
    }
  \label{fig_motivation}
  \vspace{-1.5em}
\end{figure}

To address the limitations of releasing a protected dataset, \citet{Liu2018a} suggested generating synthetic data that can be released instead of the private dataset.
Subsequently, various deep learning-based generative models for trajectory data have been developed.
Many of these approaches yield utility that substantially surpasses the traditional publication approaches.
Yet, \citet{pets24_paper} find that no existing deep learning model provides formal privacy guarantees.

As illustrated in \figref{fig_motivation}, private information can flow over two paths from the real dataset into the synthetic data.
First, during training, the model updates its parameters to learn the real data distribution.
While the specifics depend on the model and loss function, the parameters must encode some real data information to generate meaningful results.
Second, some models use input derived from the real dataset, referred to as \textit{conditional information}, which guides generation~\cite{DiffTraj23, Rao2020}.
This input can represent derived attributes~\cite{DiffTraj23}, such as average speed or start location, or entire trajectories~\cite{Rao2020}.
Although few unconditional models exist~\cite{DiffTraj23}, most \gls{sota} generative trajectory models rely on such conditional information to yield meaningful utility~\cite{pets24_paper}.

Due to the potential of deep learning-based generative trajectory models, this study addresses the overarching research question:
  \mainRQ?
Since \gls{dp} serves as the current de facto standard for privacy, our analysis focuses on this technique to ensure formal guarantees.
While non-\gls{dl} approaches for \gls{dp} trajectory generation exist, these have already been thoroughly reviewed by multiple existing studies~\cite{pets24_paper,Miranda-Pascual2023,Kapp2023_gen,Jin2021}.
Therefore, this study explicitly focuses on deep learning-based trajectory generation which has not been comprehensively evaluated regarding formal privacy guarantees before.
Following \figref{fig_motivation}, the study seeks to answer the research question by examining the impact of three factors:

\ballnumber{1}
First, we assess the impact of formal guarantees \wrt the training data on utility.
As \gls{dp-sgd} is the most widely used method for providing \gls{dp} guarantees \wrt a \gls{dl} model's training data~\cite{dpfyML}, we pose the following question:
\begin{research_question}
  \item \rqDPSGD \label{rq_dpsgd}
\end{research_question}
We find that using \gls{dp-sgd} with a realistic privacy parameter ($\varepsilon = 10$) leads to a significant utility loss, which is restrictive for most practical applications, especially on smaller datasets.
However, comparison with the \gls{sota} baseline PrivTrace~\cite{PrivTrace2023} shows that this loss is in line with non-\gls{dl} approaches.

\ballnumber{2}
For unconditional models without additional input, \gls{dp-sgd} ensures formal privacy guarantees \wrt the generated data.
However, most \gls{sota} models perform significantly better with conditional information~\cite{DiffTraj23,DiffRnTraj24,Rao2020}.
Directly using \condInfo during generation poses a critical privacy risk, as even statistical data can leak private information and violate \gls{dp} guarantees.
For instance, LSTM-TrajGAN has been shown to converge towards an identity function when overtrained~\cite{pets24_paper}.
However, to the best of our knowledge, no trajectory-level \gls{dp} approach for conditional information has been proposed, leading to the second research question:
\begin{research_question}
  \item \rqDPCond \label{rq_dpcond}
\end{research_question}
To tackle this, we introduce a novel \gls{dp} embedding for conditional information, combining compression and clipping with a \gls{dp} mechanism and prove its guarantees formally.
This \gls{dp} \condInfo can enhance utility and stabilise training while providing guarantees in some settings, especially for \glspl{gan}, combined with \gls{dp-sgd}, and on smaller datasets.

\ballnumber{3}
While diffusion models consistently outperform other approaches for trajectory generation without privacy guarantees~\cite{DiffTraj23, DiffRnTraj24, Map2Traj2024}, other architectures, such as \glspl{gan} and \glspl{vae}, can be adapted with minimal modifications.
The reduced training time of \glsshortpl{vae} or the distinct structure of \glsshortpl{gan} may influence the utility-privacy trade-off under formal guarantees, motivating the third research question:
\begin{research_question}
  \item \rqArchitecture \label{rq_architecture}
\end{research_question}
Specifically, we aim to determine whether the best-performing non-\gls{dp} model also remains the best-performing \gls{dp} model or if trends shift when applying formal guarantees.
To address this, we compare a diffusion model, a \glsshort{vae}, and a \glsshort{gan}, all derived from a similar UNet architecture.
Aligning with the findings of related work~\cite{DiffTraj23}, the conditional diffusion model yields the best performance without formal guarantees, although the \gls{vae} achieves competitive results especially on the smaller dataset.
However, when considering full formal guarantees, the \gls{vae} cannot compare with the diffusion model while the \gls{gan} performs even superior.

\subheading{Contributions}To provide a comprehensive evaluation, we follow the framework of \citet{pets24_paper} and evaluate on two public datasets, namely Porto~\cite{porto_taxi} and GeoLife~\cite{GeoLife1}, across eleven metrics, covering all four types of trajectory utility~\cite{pets24_paper}.
This study contributes to trajectory privacy by:
\begin{enumerate}[label=\textbf{\arabic*.}]
  \item Quantifying the \textit{utility cost of \glsshort{dp-sgd}} on \glsshort{sota} trajectory generation models.
  \item Proposing a method to utilise \textit{conditional information without violating differential privacy}, proving its privacy guarantees, and evaluating its impact on utility.
  \item Comparing the utility-privacy trade-off of \textit{three model types}: Diffusion, \glsshort{vae}, and \glsshort{gan}.
  \item Conducting a comprehensive evaluation using \textit{eleven metrics} and \textit{two datasets} including one \gls{sota} non-\gls{dl} baseline~\cite{PrivTrace2023}.
  \item Publishing all code and data to support further research\footnote{\label{ref_code}
    \ifanonymous
      \href{https://anonymous.4open.science/r/ptg_anon/}{Anonymous GitHub} or
      \href{https://figshare.com/s/df086bd8f78d26a70b41}{figshare.com}
    \else
      \url{https://github.com/erik-buchholz/CostOfTrajectoryPrivacy}
    \fi
  }.
\end{enumerate}

\section{Background}\label{sec_background}

This section outlines relevant background knowledge.
It defines the trajectory dataset (\secref{bg_trajectory}), introduces differential privacy (\secref{bg_dp}), and reviews generative models (\secref{bg_generative_models}) and \gls{dp-sgd} (\secref{bg_dpsgd}).

\subsection{Trajectory}\label{bg_trajectory}

A trajectory dataset is a set of \gls{n} trajectories $D = \left\{ T_1, T_2, \ldots, T_n \right\}$.
Each trajectory $T_i$ represents an ordered sequence of points:
\[
  T_i = \left( p^{(i)}_1, p^{(i)}_2, \ldots, p^{(i)}_L \right)
\]
where $L$ is the length of the trajectory.
In this work, we assume that each point $p^{(i)}_j$ consists of at least two location coordinates, typically geographical coordinates such as latitude and longitude:
$p^{(i)}_j = \left( lat, lon \right)$.
Trajectories may include temporal or semantic context~\cite{Tu2019}, but this work focuses solely on location information.

\subsection{Differential Privacy}\label{bg_dp}

Privacy notions are commonly classified as \textit{syntactic} or \textit{semantic}~\cite{Majeed2023}.
\glsreset{dp}\textit{\gls{dp}}~\cite{Dwork2013} represents the primary semantic notion for protecting personal data~\cite{Guerra-Balboa2022,Majeed2023}.
Differential Privacy ensures that the inclusion or exclusion of a single user's data does not significantly affect the output, offering \emph{plausible deniability} regarding participation in the dataset.
Formally, a mechanism satisfies $(\varepsilon, \delta)$-DP if~\cite{Dwork2013}:

\begin{definition}[Differential Privacy~\cite{Dwork2013}]
	A mechanism $\gls{mechanism}$ provides $(\varepsilon, \delta)$-differential privacy if for all \textit{adjacent} datasets $D_1$ and $D_2$ ($\gls{adj}(D_1, D_2)$), and all $S \subseteq Range(\gls{mechanism})$:
	\begin{equation}\label{eq:dp}
		\gls{prob}\left[\gls{mechanism}(D_1)\in S\right] \leq e^{\varepsilon}\times \gls{prob}\left[\gls{mechanism}(D_2)\in S\right] + \delta
	\end{equation}
\end{definition}

The main privacy parameters in differential privacy are $\varepsilon$ and $\delta$.
Typical values for $\varepsilon$ range from \num{0.01} to \num{10}, with $\varepsilon \leq 1$ considered strong and $\varepsilon \leq 10$ realistic~\cite{Erlingsson2014}.
In \gls{dl}, larger $\varepsilon$ values are often used in practice, as they may still offer empirical protection~\cite{dpfyML}.
Best practice is to set $\delta \ll \nicefrac{1}{n}$, e.g., $\delta = \nicefrac{1}{n^{1.1}}$, where $n$ is the dataset size~\cite{dpfyML}.

Differential Privacy requires defining when two datasets are \textit{neighbouring} (\gls{adj}).
The two most common adjacency relations are \textit{add-or-remove} and \textit{replace-one} (or \textit{substitution})~\cite{dpfyML}.
Add-or-remove treats datasets as adjacent if one can be obtained from the other by adding or removing a single record.
Replace-one assumes one record is replaced, keeping the dataset size fixed.
The replace-one relation is approximately twice as strong, as it combines both adding and removing a record~\cite{dpfyML} (e.g., $\varepsilon = 2$ under replace-one is roughly as strong as $\varepsilon = 1$ under add-or-remove).
While add-or-remove is more common, replace-one enables the analysis of mechanisms requiring a fixed number of records.
Moreover, the \gls{uop}~\cite{pets24_paper,dpfyML} defines what constitutes a single record.
While standard \gls{dp} aims for \textit{user-level} privacy, we adopt the \textit{trajectory-level} notion, treating each trajectory as a record~\cite{pets24_paper}.
This is common in \gls{dl}, where guarantees are typically applied per training sample, reducing utility loss at the cost of weaker privacy guarantees~\cite{dpfyML}.

Additionally, \gls{dp} provides \textit{composition theorems} that enable combining multiple mechanisms while preserving formal guarantees.
For example, it is invariant to post-processing and supports both sequential and parallel composition~\cite{Dwork2013}.

\gls{dp} can be achieved through various mechanisms:

\subheading{Laplace Mechanism}\label{bg_laplace}
The \textit{Laplace mechanism}~\cite{Dwork2013} is commonly used for continuous values.
It adds \gls{iid} noise drawn from $\gls{lap}(\Delta_1 f / \varepsilon)$ to each component of a function $f(X)$, where $\Delta_1 f$ is the $l_1$ sensitivity, i.e., the maximum change in $f$ between any two adjacent datasets. 
This yields a mechanism that satisfies $\varepsilon$-\gls{dp}.

\subheading{Gaussian Mechanism}\label{bg_gaussian}
The Gaussian mechanism adds noise based on the $l_2$-sensitivity $\Delta_2 f$, making it preferable for high-dimensional functions $f$.
It samples from a normal distribution $\mathcal{N}(0, \sigma^2)$, where $\sigma = \nicefrac{\Delta_2 f}{\varepsilon} \cdot \sqrt{2 \ln(1.25/\delta)}$ ensures $(\varepsilon, \delta)$-differential privacy for $\varepsilon, \delta \in (0,1)$~\cite{Dwork2013}.
\citet{Balle2018} propose the analytical Gaussian mechanism, which applies to any $\varepsilon > 0$ and provides tighter bounds.

\subheading{VMF Mechanism}\label{bg_vmf}\glsreset{vmf}
The \textit{\gls{vmf}} mechanism~\cite{Weggenmann2021} provides \gls{dp} guarantees by sampling from a \gls{vmf} distribution, defined over unit vectors $v \in \mathbb{R}^d$ with $\|v\|_2 = 1$.
Its \glsshort{pdf} is
\begin{equation}
  \gls{V}(\kappa, \mu)(x) = C_d(\kappa) e^{\kappa \mu^T x}
\end{equation}
where $\kappa > 0$ is the concentration parameter, $\mu$ is the mean direction, and $C_d(\kappa)$ is a normalisation constant.
The mechanism satisfies $(\varepsilon, 0)$-\gls{dp} when $\kappa = \nicefrac{\varepsilon}{\Delta_2 f}$.

\subsection{Generative Models}\label{bg_generative_models}

Generative models for sequential data~\cite{Eigenschink2023}, such as trajectories, address limited data availability and enable synthetic data generation for sensitive domains.

\textbf{\glsfirstpl{ae}}~\cite{autoencoders}\label{bg_ae} consist of an encoder that maps inputs to a lower-dimensional latent representation and a decoder that reconstructs the original input.
The latent bottleneck forces the outputs to differ, while the most important information is preserved.

\textbf{\glsfirstpl{vae}}~\cite{vae2013}\label{bg_vae} extend \glspl{ae} by encoding inputs into a latent distribution with mean $\mu$ and standard deviation $\sigma$. 
The decoder samples from this distribution to reconstruct the input.
The loss combines a reconstruction loss and the KL divergence, encouraging the latent space to converges towards $\mathcal{N}(0, 1)$.
For reconstruction, the decoder samples from the distribution defined by the encoder, while for generation, it samples from the standard normal distribution $\mathcal{N}(0, 1)$ generating new unseen samples.

\textbf{\glsfirstpl{gan}}~\cite{Goodfellow2014}\label{bg_gan} consist of a generator $G$ and a discriminator $D$.
The generator $G$ transforms input noise $z \sim p_z$ into samples resembling real data, while the discriminator $D$ evaluates whether a sample is real or generated. 
The generator and discriminator train iteratively until $G$ produces samples that are indistinguishable from real data (for the discriminator).

\textbf{Diffusion Models}~\cite{diffusion2020}\label{bg_diffusion} generate data by denoising Gaussian noise through a learned reverse diffusion process.
Starting from $\hat{x}_T \sim \mathcal{N}(0, I)$, the model iteratively predicts and removes noise to recover $\hat{x}_0$.
During training, noise is incrementally added to real data $x_0$ in the \textit{forward diffusion process}, effectively converting the sample into Gaussian noise. 
UNet~\cite{unet2015} architectures are commonly used.
Diffusion models produce high-quality outputs with stable training but are slower to sample compared to \glspl{gan} or \glspl{vae}.  %

\subsection{Differentially Private SGD}\label{bg_dpsgd}

\gls{dl} models are often trained on sensitive data containing personal or proprietary information.
\glspl{mia}~\cite{Shokri2017} have shown that such models can leak training data, even if the data itself is not published.
To mitigate this, \gls{dp} has been applied to \gls{dl}~\cite{Abadi2016,dpfyML}.
We refer to \citet{dpfyML} for a comprehensive overview and focus here on \textit{\gls{dp-sgd}}~\cite{Abadi2016}, the most widely used method to achieve \gls{dp} in \gls{dl}.
The core idea behind \gls{dp-sgd} is the addition of Gaussian noise to the gradients, which are used to update the weights of the model.
The required noise level depends on the sensitivity $\Delta f$ of the mechanism (\refer \secref{bg_dp}), which in \gls{dl} corresponds to the influence of a single training sample.

Since gradient norms are generally unbounded, \gls{dp-sgd} first clips per-sample gradients to a norm bound $C$.
Second, Gaussian noise is added to the gradients based on this clipping norm $C$ and a noise multiplier $\sigma$.
A privacy accountant tracks the accumulated privacy loss during training, yielding a final $(\varepsilon, \delta)$-\gls{dp} guarantee.
The noise multiplier can be estimated in advance to target a specific privacy budget.

\section{Threat Model}\label{sec_threat_model}

As shown in \figref{fig_motivation}, synthetic data generation can leak private information during \textit{training} of the generative model (\ballnumber{1}) or \textit{generation} of synthetic data (\ballnumber{2}).
Therefore, we distinguish four threat models based on the guarantees during training and generation.
We consider a private dataset $D$, partitioned into disjoint training (\gls{D_train}) and test (\gls{D_test}) sets, where \gls{D_test} provides conditional information for generation.

This work targets trajectory-level guarantees~\cite{pets24_paper}, where each sample corresponds to a single trajectory.
While user-level guarantees are stronger, sample-level guarantees are more common in deep learning~\cite{dpfyML} and offer a reasonable trade-off between utility and privacy (\refer \secref{bg_dp}).

\ballnumber{A}\subheading{No Guarantees}
At this level, no formal privacy guarantees are provided.
Privacy relies solely on the synthetic data's inherent properties and assumes that used conditional information does not leak sensitive details.
For instance, using statistical features as done by DiffTraj~\cite{DiffTraj23}, may be acceptable in low-risk scenarios.
Most \gls{sota} models~\cite{DiffTraj23, Rao2020, pets24_paper} operate in this setting without formal guarantees.

\ballnumber{B}\subheading{Training-level Guarantees}
Guarantees apply only to the training data \gls{D_train}, offering protection against attacks such as \gls{mia}~\cite{Shokri2017}.
This setting is relevant when conditional information is non-sensitive or the risk of leakage during generation is low.
For example, a conditional model could use synthetic or public conditional data.

\ballnumber{C}\subheading{Generation Guarantees}
This level ensures privacy only during generation \wrt the test set \gls{D_test}.
Unconditional models inherently satisfy this, as they use no input during generation, yielding $0$-\gls{dp} \wrt \gls{D_test}.
This setting is relevant when model or output leakage is unlikely, or if training uses public data and generation uses private data.

\ballnumber{D}\subheading{Full Guarantees}
This level combines \ballnumber{B} and \ballnumber{C}, offering formal guarantees for both the training data and any \condInfo used during generation.
While providing the strongest protection, it may lead to significant utility loss.

\section{Related Work}\label{sec_related}

The potential value of trajectory data for various applications has driven extensive research on privacy-preserving methods for sharing trajectories.
\citet{Miranda-Pascual2023} provide a detailed review of \gls{dp} trajectory publishing approaches.
Their work highlights the shortcomings of the existing approaches and motivates the need for further research.
\update{
  As one state-of-the-art representative for non-deep learning approaches, we include PrivTrace~\cite{PrivTrace2023} (\refer \secref{sec_setup}), for its strong utility under \gls{dp} and available source code.
  PrivTrace generates synthetic trajectories in three steps: \begin{enumerate*}
  \item discretising the dataset's space with a density-aware grid,
  \item training first- and second-order Markov models with \gls{dp}, and
  \item generating trajectories via a random walk algorithm.
  \end{enumerate*}
}

As an alternative to traditional methods, \citet{Liu2018a} propose deep learning-based generative models for trajectory generation.
\citet{pets24_paper} systematise \gls{sota} deep generative models for trajectory data and introduce an evaluation framework.
They show that these models offer strong utility but lack formal privacy guarantees.
They also identify privacy risks in some models and recommend \gls{dp-sgd} to mitigate them.
However, \gls{dp-sgd} requires models to avoid additional inputs during generation, unlike many \gls{sota} models such as LSTM-TrajGAN~\cite{Rao2020}.
Recently, diffusion models have emerged that do not rely on such inputs.

\textit{DiffTraj}~\cite{DiffTraj23} is the first application of the \gls{ddpm} diffusion model for trajectory generation and serves as the baseline in this work.
Its architecture is based on a UNet~\cite{unet2015} with up- and downsampling modules built from ResNet blocks using \gls{conv1d} layers.
An attention-based transition module~\cite{Vaswani2023} connects these components to capture sequential dependencies.
The model supports optional \textit{conditional information} for guided generation, such as average speed, distance, and start time.
Despite these strengths, DiffTraj lacks formal privacy guarantees, and while diffusion models are more robust to \gls{mia} than \glspl{gan}, they remain vulnerable~\cite{Matsumoto2023}.

\textit{Diff-RNTraj}~\cite{DiffRnTraj24} generates trajectories directly on the road network, avoiding errors from generating GPS trajectories followed by map-matching.
A spatial validity loss penalises unconnected segments, improving performance over baselines.
However, it requires road network data and lacks formal privacy guarantees.
We exclude Diff-RNTraj from our experiments due to these constraints.

To the best of our knowledge, ConvGAN~\cite{convgan} is the only approach applying \gls{dp-sgd} to trajectory generation.
While ConvGAN represents the first proof-of-concept for this methodology, the resulting model cannot generate trajectories with real-world utility, both with and without \gls{dp} guarantees~\cite{convgan}.
Therefore, we do not further consider ConvGAN.

\section{Methodology}\label{sec_design}

This section outlines the methodology to address the research questions defined in \secref{sec_intro}.
As depicted in \figref{fig_design_overview}, information flows from the real dataset \gls{D_real} to the synthetic dataset \gls{D_syn} along two paths.
First, during training model parameters are updated using \gls{D_train} (blue path).
To ensure privacy \wrt the training data, we apply \gls{dp-sgd}~\cite{Abadi2016} (\secref{sec_method_dpsgd}).
Second, during generation (dotted path), the model optionally incorporates conditional information derived from \gls{D_test}.
To protect this input, we propose a novel \gls{dp} conditional embedding mechanism (\secref{sec_conditional_info}).
The overall design aims to ensure formal guarantees along both paths and quantify the resulting utility-privacy trade-off.
To generalise findings and assess the impact of model choice, we evaluate three generative model types -- diffusion, \gls{gan}, and \gls{vae} -- detailed in \secref{sec_architecture}.
The privacy analysis of the proposed design is provided in \secref{sec_privacy_analysis}.

\begin{figure}
  \centering
  \includegraphics[width=\linewidth]{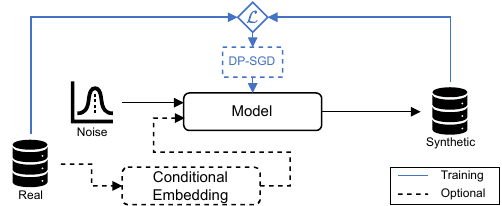}
  \vspace{-1.5em}
  \caption{
  The dataset is connected to the synthetic data through two paths that leak information.
  First, during training, the model parameters are updated.
  Here, \gls{dp-sgd} can ensure privacy.
  Second, conditional embedding is optionally derived from real samples and used during generation.
  }
  \label{fig_design_overview}
  \vspace{-1.5em}
\end{figure}

\subsection{Privacy of Training}\label{sec_method_dpsgd}
During training, model parameters are updated using gradients derived from a loss.
This loss reflects real data either directly, e.g., via the reconstruction loss in a \gls{vae}, or indirectly, e.g., via the discriminator in a \gls{gan} (\refer \secref{bg_generative_models}).
Thus, model parameters may encode aspects of the real data necessary for generating realistic synthetic samples (\refer blue path in \figref{fig_design_overview}).
This risks memorising and reproducing training samples (via overfitting), and allows \glspl{mia}~\cite{Shokri2011} to infer details about \gls{D_train}.
The most widely used method to prevent such leakage is \gls{dp-sgd}\cite{Abadi2016,dpfyML}.

We use the Opacus library~\cite{Opacus} with \texttt{PRV} accounting~\cite{prv_accountant2021} to implement \gls{dp-sgd}.
The noise multiplier is pre-determined based on the desired target values for $(\gls{eps_s}, \gls{delta_s})$, with $\gls{eps_s} = 10.0$ as default.
This value the largest value still considered offering realistic guarantees in \gls{dl}~\cite{dpfyML}.
Smaller values incur a high utility cost, while larger ones are too weak for formal guarantees~\cite{dpfyML}, although they may still protect against specific attacks.
We follow the standard recommendation $\gls{delta_s} = \nicefrac{1}{n^{1.1}}$ with $n = |\gls{D_train}|$~\cite{dpfyML}.

\gls{dp-sgd} is applied to all gradients for the diffusion model and \gls{vae}.
For the \gls{gan}, only the generator is trained with \gls{dp}, while the discriminator is trained without it to preserve training stability.
This does not compromise privacy \wrt the synthetic data (or release of the generator), as the discriminator is not involved in generation.
However, the discriminator must not be released.

\subsection{Conditional Information}\label{sec_conditional_info}

Most \gls{sota} generative trajectory models, such as DiffTraj~\cite{DiffTraj23} and LSTM-TrajGAN~\cite{Rao2020}, rely on conditional information for high utility.
For instance, LSTM-TrajGAN uses encoded trajectories, while DiffTraj incorporates eight derived features.
This creates an information flow from the real dataset \gls{D_real} to the generated dataset \gls{D_syn} (\refer \ballnumber{2} in \figref{fig_motivation}).
As a result, conditional models lack formal privacy guarantees, even when trained with \gls{dp-sgd}.
The risk of sensitive information leaking through conditional inputs has been demonstrated for LSTM-TrajGAN~\cite{pets24_paper}. 
To mitigate this risk, we introduce a novel \gls{dp} conditional embedding mechanism that ensures privacy \wrt the \condInfo used during generation.
It is illustrated in \figref{fig_cond_info} and marked by the black dotted box in \figref{fig_design_overview}.
We describe its components in the order they are applied.

\begin{figure}
  \centering
  \includegraphics[width=\linewidth]{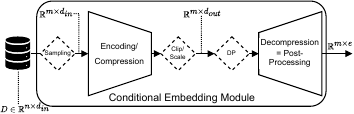}
  \vspace{-1.5em}
  \caption{
    First, $m$ samples from the dataset are selected uniformly at random.
    A neural network maps these samples to a lower-dimensional space $d_{out}$.
    Optionally, the per-sample norm is clipped/scaled, and a \gls{dp} mechanism is applied.
    Finally, the resulting encoding is decompressed.  %
  }
  \label{fig_cond_info}
  \vspace{-1em}
\end{figure}

\subheading{1) Input Format}
We consider two input formats of conditional information.
The first format, referred to as \textbf{\textit{\DTcond}}, corresponds to the original DiffTraj~\cite{DiffTraj23} format and includes:
\begin{enumerate*}
  \item departure time,
  \item total distance,
  \item total duration,
  \item total length,
  \item average distance,
  \item average speed,
  \item start grid ID,
  \item end grid ID.
\end{enumerate*}
This input is embedded using the original wide-and-deep network~\cite{DiffTraj23}.
If the \gls{dp} mechanism is not enabled (threat models \ballnumber{A} and \ballnumber{B}), the resulting embedding is used directly and the following steps 2--6 are skipped.

The second format, \textbf{\textit{\Sample}}, uses a full trajectory sample from the test set \gls{D_test} as conditional input.
Naturally, using a trajectory as conditional information suffers from the risk of the models returning the trajectory unaltered, as can happen with LSTM-TrajGAN~\cite{pets24_paper}.
Therefore, we never use this format without applying the \gls{dp} mechanism, i.e., all steps 1--6 are mandatory.

\subheading{2) Sampling}
In the first step, we uniformly sample $m \leq |D_{test}|$ trajectories \update{without replacement} from $D_{test}$.
This step is crucial for the privacy guarantees, and the value of $m$ must be fixed upfront independently of the dataset to prevent information leakage.

\subheading{3) Compression}
The amount of noise added to achieve \gls{dp} guarantees depends on the dimensionality of the data.
Therefore, we compress the $m$ samples to a lower-dimensional space $d_{\text{out}}$ using a single \gls{fc} layer.

\subheading{4) Clipping and Scaling}
To apply a \gls{dp} mechanism, the sensitivity of each sample must be bounded (\refer \secref{bg_dp}).
We achieve this by clipping or scaling the compressed embedding $d$ to a fixed norm $C$.
The norm type $p$ depends on the \gls{dp} mechanism used.
For the Laplace ($p=1$) and Gaussian ($p=2$) mechanisms, we apply standard norm clipping as in \gls{dp-sgd}~\cite{Abadi2016}:
\begin{equation}
  d_{\text{clipped}} = \min\left(1, \frac{C}{\|d\|_p}\right) \cdot d
\end{equation}
For the \gls{vmf} mechanism ($p=2$), requiring values on the unit sphere, we instead rescale all vectors to norm $C$~\cite{Faustini2022}:
  \( d_{\text{scaled}} = \frac{C}{\|d\|_2} \cdot d \).
Empirically, we observe that compression tends to fully utilise the available norm range.
Therefore, we fix $C = 1.0$, which is compatible with all \gls{dp} mechanisms.

\subheading{5) DP Mechanism}
We consider three \gls{dp} mechanisms for the conditional embedding module: Laplace \gls{lap}, Gaussian \gls{gau}, and \gls{vmf} \gls{V}.
Unless stated otherwise, we use the Laplace mechanism, which performed best in preliminary experiments.
The others variants are evaluated in ablation study III in \secref{eval_mechanisms}.
The Laplace and Gaussian mechanisms add \gls{iid} noise to each component of the clipped embedding:
\begin{equation}
  d_{\text{DP}} = d_{\text{clipped}} + \text{Noise}(\sigma)
\end{equation}
with noise scale $\sigma$ determined by the clipping norm $C$ and privacy parameters \gls{eps_c} (and \gls{delta_c} for \gls{gau}).
The \gls{vmf} mechanism replaces $d_{\text{scaled}}$ with a unit vector sampled from the \gls{vmf} distribution, using $d_{\text{scaled}}$ as the mean direction and a concentration parameter $\kappa$.
The values for $\sigma$ and $\kappa$ are detailed in the privacy analysis in \secref{sec_privacy_conditional}.

\subheading{Noise Schedule}
During training, we apply a noise schedule similar to the $\beta$-schedule used in diffusion models~\cite{DDPM2021}.
For each sample in a batch, we draw a random value $\beta \sim \gls{uni}(0, 1)$ and compute the final embedding as:
\begin{equation}
  d_{\text{DP}} = \beta \cdot d_{\text{DP}} + (1 - \beta) \cdot d_{\text{clipped}}.
\end{equation}
This approach allows the model to occasionally use less noisy embeddings, which helps stabilise training, while still adapting to the noise level used during generation.
The noise schedule \textbf{must not} be applied during generation, as this would violate the privacy guarantees.
However, since the training data is protected via \gls{dp-sgd}, applying the schedule during training does not compromise privacy \wrt $\gls{D_train}$.

\subheading{6) Decompression}
After the application of the \gls{dp} mechanism, the noisy embedding $d_{\text{DP}}$ is decompressed using one \gls{fc} layer to the embedding dimension $e$ used by the model.
We fix $e = 512$ according to the value used by DiffTraj~\cite{DiffTraj23}.

\subheading{Summary}\label{sec_types_of_cond_info}
 We consider six types of conditional embeddings:
\begin{enumerate}[label=\arabic*)]
  \item \textbf{None:} Unconditional models.
  \item \textbf{\DTcond:} Original DiffTraj~\cite{DiffTraj23} embedding.
  \item \textbf{\DTcond + \gls{lap}:} DiffTraj embedding with Laplace mechanism.
  \item \textbf{\Sample + \gls{lap}:} Full sample with Laplace mechanism.
  \item \textbf{\Sample + \gls{gau}:} Full sample with Gaussian mechanism.
  \item \textbf{\Sample + \gls{V}:} Full sample with \gls{vmf} mechanism.
\end{enumerate}
\Sample encoding without \gls{dp} is not considered, as this case does not provide any privacy and suffers the same risk as LSTM-TrajGAN~\cite{pets24_paper}.
The privacy guarantees of this conditional embedding module are proven in \secref{sec_privacy_conditional}.

\subsection{Model Architectures}\label{sec_architecture}

We consider three generative model types: Diffusion, \gls{gan}, and \gls{vae}, each with an \textit{unconditional} and \textit{conditional} variant.
The unconditional models follow standard structures (\secref{bg_generative_models}), while the conditional variants include the embedding module described in \secref{sec_conditional_info}.
Model architectures were kept similar to isolate the effect of model type.

\subheading{Diffusion}
As a diffusion model, we use the \gls{sota} DiffTraj~\cite{DiffTraj23} model described in \secref{sec_related}.
The only difference is the addition of \gls{dp-sgd} and the integration of the conditional embedding module outlined in the previous sections.
The DiffTraj model is based on the UNet~\cite{unet2015} architecture published by Rombach et al.~\citep{unet_code2021, unet_code2021a}.
The proposed \textit{Traj-UNet} consists of a downsampling and an upsampling module connected by a transition module.
Both the down- and upsampling modules consist of multiple stacked ResNet blocks based on \gls{conv1d} layers.
Moreover, each ResNet block receives the sum of the diffusion step's $t$ embedding and the conditional embedding (if used) as secondary input.
The downsampling blocks connect to the equivalent upsampling blocks through skip connections.
The transition module uses an attention mechanism~\cite{Vaswani2023} between two ResNet blocks to capture sequential dependencies.

\subheading{VAE}
For the \gls{vae}, we split the Traj-UNet into two parts, using the UNetEncoder and UNetDecoder~\cite{unet_code2021, unet_code2021a}.
The downsampling module represents the \gls{vae}'s encoder, while the upsampling module represents the decoder.
However, to enable both modules to capture sequential dependencies, we include the attention-based transition module both at the end of the encoder and the beginning of the decoder.
As we split the Traj-UNet into two separate modules, no skip connections are used.
Moreover, the ResNet blocks only receive the \condInfo as secondary input, the encoder returns the latent distribution parameters $\mu$ and $\sigma$ (\refer \secref{bg_vae}), and the decoder applies a \texttt{tanh} activation to the output.

\subheading{GAN}
For the \gls{gan}, we initially tested using a full Traj-UNet for both generator and discriminator, but the increased complexity caused unstable training.
Therefore, we used the UNetDecoder, also used by the \gls{vae}, as the generator and the UNetEncoder as the discriminator.
Two layers, a \gls{conv1d} and a \gls{fc} layer, were added to the discriminator to reduce the output to a single value.
Using a full Traj-UNet in only one of the modules also yielded poor results due to imbalance.
Since the standard adversarial loss caused unstable training, we used the \gls{wgan-lp} loss~\cite{wgan_lp} for better stability, despite slower convergence.
For non-\gls{dp-sgd} training, we set $n_{\text{critic}} = 5$ to maintain a near-optimal discriminator, as $n_{\text{critic}} = 1$ performed poorly.
In the evaluation, we report the total training steps, resulting in only \num{20 000} generator updates when using $n_{\text{critic}} = 5$.
This setup ensures a runtime comparable to the diffusion model for the same total number of updates.
Under \gls{dp-sgd}, we set $n_{\text{critic}} = 1$ as the added noise sufficiently weakened the generator, making additional discriminator updates unnecessary.
Thus, in this setting, generator updates equal the total steps.

\section{Privacy Analysis}\label{sec_privacy_analysis}

This section analyses the privacy guarantees of the proposed design.
As outlined in \secref{sec_threat_model}, we assume that training and generation use disjoint datasets, i.e., $\gls{D_real} = \gls{D_train} \cup \gls{D_test}$ and $\gls{D_train} \cap \gls{D_test} = \emptyset$.
First, \secref{sec_privacy_dpsgd} examines the guarantees \wrt the training data \gls{D_train} achieved through \gls{dp-sgd}.
Second, \secref{sec_privacy_conditional} analyses the guarantees \wrt the conditional information derived from $D_{gen}$ during generation.
Finally, \secref{sec_privacy_combined} combines these individual guarantees to provide a comprehensive analysis of the privacy guarantees for the final synthetic dataset.

\subsection{Privacy of Training}\label{sec_privacy_dpsgd}

To provide formal guarantees \wrt the training data \gls{D_train}, i.e., the blue path in \figref{fig_design_overview}, we apply \gls{dp-sgd}~\cite{Abadi2016}, as described in \secref{sec_method_dpsgd}.
We use the \texttt{PRV}~\cite{prv_accountant2021} accounting method and denote the privacy parameters for \gls{dp-sgd} as $(\gls{eps_s}, \gls{delta_s})$.
Both threat models \ballnumber{B} and \ballnumber{D} (\refer \secref{sec_threat_model}) rely on \gls{dp-sgd} to ensure privacy \wrt the training data.

\begin{corollary}
  All models trained with \gls{dp-sgd} provide \update{trajectory-level} $(\gls{eps_s}, \gls{delta_s})$-differential privacy \gls{wrt} the training data $D_{train}$ \update{under the add-or-remove adjacency relation}.
\end{corollary}
\noindent As we use standard \gls{dp-sgd} based on an established library, we refer to the proof of the privacy guarantees to \citet{Abadi2016} and \citet{prv_accountant2021}.
Note that \gls{dp-sgd} does not provide privacy guarantees \wrt the conditional information used during generation, which we address in the following.

\subsection{Privacy Proof for Conditional Information}\label{sec_privacy_conditional}

In the following, we derive the privacy guarantees \wrt the conditional information used as a secondary input for the generative model.
These guarantees relate to the dotted path from the real data in \figref{fig_design_overview} and are relevant for threat models \ballnumber{C} and \ballnumber{D}.
As outlined in \secref{sec_threat_model}, all guarantees are defined at the trajectory level, i.e., with respect to individual trajectories in the dataset.
First, we formalise the per-sample compression function \gls{f_c} (\refer \secref{sec_conditional_info}):

\begin{definition}[Per-Sample Compression Function]
  The per-sam\-ple compression function \gls{f_c} maps an input sample $x \in \gls{reels}^{\gls{d_in}}$ to a lower-dimensional space $\gls{reels}^{\gls{d_out}}$, followed by norm clipping (or scaling for the \gls{vmf} mechanism):
  \begin{equation}
    f_c: \mathbb{R}^{\gls{d_in}} \rightarrow \gls{reels}^{\gls{d_out}}, \quad x \mapsto f(x) \cdot \min\left\{\frac{C}{\|f(x)\|_p}, 1\right\}
  \end{equation}
  where $f(x)$ is the output of a neural network, \gls{C} is the norm bound, \gls{p} the type of norm used, and $\gls{d_in}$ and $\gls{d_out}$ denote the input and output dimensions, respectively.
  Clipping ensures the output norm does not exceed $C$, while scaling (used for \gls{vmf}) sets it exactly to $C$.
\end{definition}

Here, samples represent trajectories, i.e., $x \in \mathbb{R}^{2 \times L}$, where \gls{L} denotes the sequence length.
For simplicity, we flatten the sequence such that $\gls{d_in} = 2L$.
Note that the norm type \gls{p} needs to be consistent throughout this proof, i.e., if the Laplace mechanism is used, $p=1$ must hold for all equations, while $p=2$ is required for the Gaussian and \gls{vmf} mechanisms.

Next, we extend this function to a per-dataset function $F_c$.
We consider a trajectory dataset $D = \{x_1, x_2, \ldots, x_n\}$, consisting of $n$ samples from $\mathbb{R}^{d_{in}}$.
As described in \secref{sec_conditional_info}, we select $m < n$ samples uniformly at random without replacement from the dataset for generation, where $m$ is a dataset-independent constant determined in advance.

\begin{definition}[Dataset Compression Function]
  The dataset compression function $F_c$ receives $m$ samples and maps each sample independently to a lower-dimensional space $\mathbb{R}^{m \times d_{out}}$ using the per-sample function $f_c$:
  \begin{equation}
    F_c: \mathbb{R}^{m \times d_{in}} \rightarrow \mathbb{R}^{m \times d_{out}}, D \mapsto \{f_c(x_i)\}_{i=1}^m
  \end{equation}
\end{definition}
Each sample $x'_i = f_c(x_i)$ in the output only depends on the corresponding sample $x_i$ in the input dataset, i.e., $F_c$ represents a parallel composition of $f_c$ over disjoint subsets of the dataset $D$.

\update{We target trajectory-level privacy as the \gls{uop} and use the replace-one adjacency relation to define the neighbourhood of two datasets, i.e., one dataset is obtained from the other by replacing one sample with another.
The common add-or-remove adjacency relation is not applicable here, as the output dimensionality of the compression function $F_c$ depends on the number of input samples.}
Next, we define the privacy mechanism $M^{Lap}_c$ based on the Laplace mechanism.

\begin{definition}[Laplace Compression Mechanism with Subsampling]
  The compression mechanism $M^{Lap}_c$ receives a dataset $D$, \update{subsamples $m$ samples uniformly at random without replacement},
  and adds independently drawn Laplace noise to all components of the output of the dataset compression function $F_c$:
  \begin{equation}
    M^{Lap}_c(D): \mathbb{R}^{n \times d_{in}} \rightarrow \mathbb{R}^{m \times d_{out}} D \mapsto F_c(D) + Y
  \end{equation}
  where $Y \in \mathbb{R}^{m \times d_{out}}$ is a matrix with $Y_{ij} \sim Lap(\lambda)$.
\end{definition}

\begin{theorem}[Privacy of Laplace Compression Mechanism]
  The mechanism $M^{Lap}_c$ provides \tashsays{$(\log(1 + \frac{m}{n}(e^\varepsilon - 1)), 0)$-\gls{dp}} %
  \wrt the replace-one adjacency definition for \tashsays{$\lambda = \frac{2C}{\varepsilon}$} if $p =1$.
\end{theorem}
\begin{proof}
  The result follows directly from the proof for the Laplace mechanism (\refer Theorem~3.6 in~\cite{Dwork2013}). The value for $\lambda$ corresponds to the sensitivity of the Laplace on replace-one adjacent datasets. Uniform subsampling without replacements provides amplification as per Theorem~9 in \cite{balle2018privacy}.
\end{proof}

The Gaussian~\cite{Dwork2013} or \gls{vmf}~\cite{Weggenmann2021} mechanisms can serve as a drop-in replacement for the Laplace mechanism.
For these mechanisms, one must use the $l_2$-norm instead of the $l_1$-norm (i.e., $p=2$), and additionally, the \gls{vmf} requires fixing the norm bound $C = 1$ to obtain unit vectors.
The Gaussian mechanism is useful when $d_{out}$ is large, as the $l_2$-norm scales with $\sqrt{d_{out}}$ instead of $d_{out}$.
The \gls{vmf} mechanism has the advantage that only the direction of the conditional information vector changes, but not its norm, which might benefit some settings.
We can define the Gaussian Compression Mechanism $M^{\mathcal{G}}_c$ as follows:
\begin{definition}[Gaussian Compression Mechanism with Subsampling]
  The compression mechanism $M^{\mathcal{G}}_c$ receives a dataset $D$, \update{subsamples $m$ samples uniformly at random without replacement},
  and adds independently drawn Gaussian noise to all components of the output of $F_c$:
  \begin{equation}
    M^{\mathcal{G}}_c(D): \mathbb{R}^{n \times d_{in}} \rightarrow \mathbb{R}^{m \times d_{out}} D \mapsto F_c(D) + Z
  \end{equation}
  where $Z \in \mathbb{R}^{m \times d_{out}}$ is a matrix with $Z_{ij} \sim \mathcal{N}(0, \sigma^2)$.
\end{definition}
\begin{theorem}[Privacy of Gaussian Compression Mechanism]
  For any $\varepsilon > 0$, $\delta \in (0, 1)$, the compression mechanism $M^{\mathcal{G}}_c$ provides \tashsays{$(\log(1 + \frac{m}{n}(e^\varepsilon - 1)), \frac{m}{n} \delta)$-differential privacy} \wrt the replace-one adjacency definition for $\sigma = \frac{\alpha 2C}{\sqrt{2\varepsilon}}$ if $p = 2$, where $\alpha$ is defined as per Algorithm~1 in~\cite{Balle2018}.
\end{theorem}
\begin{proof}
  The theorem follows from the proof of the Analytical Gaussian mechanism (\refer Theorem~9 in \cite{Balle2018}). Uniform subsampling without replacements provides amplification as per  Theorem~9 in \cite{balle2018privacy}.
\end{proof}

Finally, the \gls{vmf} mechanism $M^{\mathcal{V}}_c$ can be defined as:
\begin{definition}[VMF Compression Mechanism with Subsampling]
  The compression mechanism $M^{\mathcal{V}}_c$ receives a dataset $D$, \update{subsamples $m$ samples uniformly at random without replacement},
  and independently samples from a \gls{vmf} distribution for each vector:
  \begin{equation}
    M^{\mathcal{V}}_c(D): \mathbb{R}^{n \times d_{in}} \rightarrow \mathbb{R}^{m \times d_{out}} D \mapsto \left\{\mathcal{V}(\kappa, F_c(D)_i)\right\}_{i=1}^m
  \end{equation}
  where $\mathcal{V}(\kappa, x)$ denotes a sample from the \gls{vmf} distribution with concentration parameter $\kappa$ and mean direction $x$.
  Note that $C=1$ as the \gls{vmf} mechanism requires unit vectors.
\end{definition}
\begin{theorem}[Privacy of VMF Compression Mechanism]
  The mechanism $M^{\mathcal{V}}_c$ provides \tashsays{$(\log(1 + \frac{m}{n}(e^\varepsilon - 1)), 0)$-differential privacy} \wrt the replace-one adjacency definition for $C=1$ and $\kappa = \frac{\varepsilon}{2C} = \frac{\varepsilon}{2}$ if $p = 2$.
\end{theorem}
\begin{proof}
  This theorem follows from Corollary~2 in \citet{Faustini2022}, where the batches $B$ and $B'$ in the corollary correspond to our subsets of $D$ and $D'$, respectively, and uniform sampling without replacements provides amplification as per Theorem~9 in \cite{balle2018privacy}.
\end{proof}

\subsection{Combined Privacy Analysis}\label{sec_privacy_combined}

The training data $D_{train}$ influences the generated dataset solely through its impact on the model parameters.
Thus, whether conditional information used during training is protected by \gls{dp} does not affect the privacy guarantees \gls{wrt} $D_{train}$.
Accordingly, when using \gls{dp-sgd} with privacy parameters $\varepsilon_s$ and $\delta_s$, the overall privacy for the training set is $(\varepsilon_s, \delta_s)$ under the add-or-remove adjacency relation.
Without  \gls{dp-sgd}, no guarantees \gls{wrt} $D_{train}$ are provided.

The privacy \gls{wrt} $D_{test}$ is based on the privacy mechanism (\secref{sec_privacy_conditional}) used during generation.
Using a privacy mechanism during training does not influence the guarantees related to $D_{test}$.
The generation process provides $(\varepsilon_c, \delta_c)$-\gls{dp} \gls{wrt} the conditional information under the replace-one adjacency relation, if the \gls{dp} conditional information mechanism is used.
Unconditional models that do not use any input inherently provide $(0, 0)$-\gls{dp} \gls{wrt} conditional information, as the samples in $D_{test}$ are never accessed.

\update{Due to the differing adjacency relations, the privacy guarantees \gls{wrt} $D_{train}$ and $D_{test}$ cannot (trivially) be combined using the parallel composition theorem despite $D_{train}$ and $D_{test}$ being disjoint.
We note that the replace-one adjacency relation is considered to be approx.\ twice as strong as the add-or-remove adjacency relation~\cite{dpfyML} (\refer \secref{bg_dp}).
Thus, we evaluate the \gls{dp} \condInfo mechanism for both $\varepsilon_c = 10$, and $\varepsilon_c = 20$, as the latter provides comparable privacy to \gls{dp-sgd} with $\varepsilon_s = 10$.}

Referring back to the threat model (\secref{sec_threat_model}):
Setting \ballnumber{A} provides no privacy guarantees,
\ballnumber{B} achieves $(\varepsilon_s, \delta_s)$-\gls{dp} \gls{wrt} training data $D_{train}$ under the add-or-remove adjacency relation,
\ballnumber{C} ensures $(\varepsilon_c, \delta_c)$-\gls{dp} \gls{wrt} test set $D_{test}$ under the replace-one adjacency relation, and
\ballnumber{D} provides both $\left(\varepsilon_s, \delta_s\right)$-DP \gls{wrt} training data under add-or-remove adjacency and $(\varepsilon_c, \delta_c)$-DP \gls{wrt} test set under replace-one adjacency.
\update{All guarantees are defined at the trajectory level, i.e., with respect to individual samples in the dataset, as outlined in \secref{sec_threat_model}.}

\section{Evaluation}\label{sec_eval}

In this section, we report the results of our experiments to address the research questions outlined in \secref{sec_intro}.
First, we define the experimental setup in \secref{sec_setup} and introduce the evaluation metrics in \secref{sec_metrics}.
\secref{sec_configurations} describes the used configurations, and \secref{sec_results} presents the results.

\subsection{Experimental Setup}\label{sec_setup}

\subheading{System Specifications}\label{sec_hardware}
We performed the final experiments on a \gls{hpc} cluster with notes having two Intel Xeon Platinum 8452Y CPUs (72 cores in total), \SI{512}{\giga\byte} of RAM, running SUSE Linux Enterprise Server \num{15} SP4. 
For each experiment, we restricted the resources per node to one NVIDIA H100 NVL GPU (\SI{94}{\giga\byte}), 18 CPU cores, and \SI{128}{\giga\byte} of RAM.

\subheading{Datasets}\label{sec_datasets}
Two real-world trajectory datasets are used for evaluation.
Following the preprocessing steps of DiffTraj~\cite{DiffTraj23}, both datasets are resampled to a fixed length and limited to a bounding box.
The bounding box is chosen independent of the data based on geographical constraints, such as the fifth ring road for Beijing, to prevent privacy leakage.
The first dataset is \textit{Porto}\footnote{\scriptsize
  $\left(lat_{min}, lon_{min}, lat_{max}, lon_{max}\right) = \left(41.10, -8.72, 41.24, -8.50\right)$
}~\cite{porto_taxi}, containing GPS trajectories of taxis in Porto, Portugal.
This dataset consists of \num{1559209} trajectories, resampled to a fixed length of \num{100} locations per trajectory.
The second dataset is \textit{Geolife}\footnote{\scriptsize
  $\left(lat_{min}, lon_{min}, lat_{max}, lon_{max}\right) = \left(39.75, 116.19, 40.03, 116.56\right)$
}~\cite{GeoLife1}, which includes GPS trajectories from \num{182} users.
We use the preprocessed version provided in the artifacts of \citet{pets24_paper}, containing \num{69504} trajectories, resampled to a fixed length of \num{200} locations per trajectory.

\update{
  \subheading{Non-DL Baseline}
  As detailed in \secref{sec_related}, we include PrivTrace~\cite{PrivTrace2023} as a \gls{sota} non-\gls{dl} baseline.
  We use the authors' implementation\footnote{\scriptsize \url{https://github.com/DpTrace/PrivTrace}} without modification.
  It uses hard-coded parameters $K$ and $\kappa$ that differ from those in the paper.
  As we received no response to our communication attempts and the formulas in~\cite{PrivTrace2023} yield an intractable number of Markov states (i.e., the code does not terminate), we retain the implementation defaults.
  Since PrivTrace generates variable-length trajectories, we resample them to the fixed length of each dataset for a fair comparison.
}

\subheading{Evaluation Process}
First, we train each model variant on the respective training set.
We split each dataset $D$ into five folds such that each run uses a different fold for testing (\gls{D_test}) and using \SI{80}{\percent} of the dataset for training (\gls{D_train}).
Each configuration is repeated five times with different random seeds, and the results are averaged.
The number of epochs is determined such that the number of steps is larger or equal to \num{100 000}.
After training, we generate \num{3000} synthetic trajectories for each model and dataset and compare them to \num{3000} real trajectories selected randomly from the respective test dataset \gls{D_test}.
If conditional information is used, the synthetic trajectories are matched with the real trajectories from which the conditional information was derived for the trajectory-level evaluation.
Otherwise, we match the synthetic trajectories to the real trajectories using the Hungarian algorithm~\cite{Kuhn2005} and the average per-trajectory haversine distance~\cite{haversine-formula} as the cost function.

\subsection{Evaluation Metrics}\label{sec_metrics}
The selection of evaluation metrics follows the guidelines of the framework proposed by \citet{pets24_paper}, with parameters aligned to related work wherever possible.
Implementation details are provided in the code repository\footref{ref_code}.
Unless stated otherwise, distances are computed using the Haversine formula~\cite{haversine-formula} and reported in metres.

To evaluate how well the generated data matches the overall \textit{point distribution}, we use the \textit{\gls{swd}}.
This metric compares the full 2D distribution of points directly, without requiring to discretise the data into a grid first.
For comparability with existing methods~\cite{DiffTraj23}, we also report (i) the \textit{\gls{jsd}} after discretising the space into a $64 \times 64$ grid, and (ii) the \textit{\gls{hd}} between \num{100000} randomly sampled real and synthetic points.

\textit{Point statistics} evaluate how well local spatial patterns are preserved in the synthetic data.
(i) \textit{Range queries}~\cite{Chen2011} estimate local density by counting how many points lie within a radius $r$ of a query location.
We report the \textit{\gls{mre}} over \num{200} random queries with radii $r \in \{50, 100, 200, 500, 1000\}\unit{\meter}$.
(ii) \textit{Hotspot analysis}~\cite{Cunningham2022} checks whether both datasets identify the same high-density regions.
The space is divided into a $g\times g$ grid for $g = 128$, and cells exceeding the $95$th percentile in density are marked as hotspots.
The \textit{\gls{sdc}} quantifies the overlap between real and synthetic hotspot sets.

At the \textit{trajectory level}, we assess both preservation and statistical properties.
(i) \textit{Trajectory preservation} measures how closely synthetic trajectories replicate the shape and location of real ones.
Reported metrics include the per-trajectory \textit{\gls{hd}}, the \textit{normalised Haversine distance} (divided by trajectory length), and the \textit{\gls{dtw}}\cite{dtw2007}, which aligns points between two trajectories to minimise cumulative distance.
(ii) \textit{Trajectory statistics} describe global mobility behaviour, including the \textit{\gls{ttd}}\cite{Sun2023} (total length of a trajectory), the \textit{trajectory diameter}\cite{Sun2023} (maximum distance between any two points), and the \textit{trip error}~\cite{Sun2023}, which compares start-end patterns via transition probabilities between grid cells (using a $g\times  g$ grid for $g=16$).
For the \gls{ttd} and diameter, histograms are constructed using \num{55} bins (based on the rule of thumb $\sqrt{3000} \approx 55$) and compared using the \textit{\gls{jsd}}.
The trip error is computed by flattening the transition matrices $U \in g^2 \times g^2$ and comparing them via \textit{\gls{jsd}}.

\subsection{Configurations}\label{sec_configurations}

To answer the three research questions, we vary the following parameters in our evaluation:
\begin{enumerate*}[label=\Roman*)]
  \item First, we train each base configuration with and without \gls{dp-sgd} to assess its impact across different settings.
  \item Second, we evaluate all six types of conditional information described in \secref{sec_types_of_cond_info}, namely:
  \begin{enumerate*}[label=\arabic*)]
    \item no conditional information,
    \item the DiffTraj embedding (\DTcond),
    \item \DTcond with Laplace noise (\gls{lap}),
    \item the proposed \Sample embedding with Laplace noise (\gls{lap}),
    \item Gaussian noise (\gls{gau}), and
    \item \gls{vmf} noise (\gls{vmf}).
  \end{enumerate*}
  \item Third, we test the three model architectures from \secref{sec_architecture} -- DiffTraj, UNetVAE, and UNetGAN -- to analyse the effect of the architecture on performance.
\end{enumerate*}

Unless stated otherwise, we use $\gls{eps_s} = 10$ and $\delta = \nicefrac{1}{n^{1.1}}$ for \gls{dp-sgd}, as reasoned in \secref{sec_method_dpsgd}.
The maximum gradient norm for \gls{dp-sgd} is set to $C = 0.1$.
For conditional information, we use $\gls{eps_c} = 10$ and compute $\delta$ analogously based on the test set size.
Since privacy amplification depends on the number of samples $m$, we conservatively assume $m = n$, i.e., the synthetic dataset matches the test set size.
If fewer samples are generated (as in our evaluation), the actual guarantee is stronger than the reported $\gls{eps_c}$ due to amplification by subsampling.
The dimensionality of compressed conditional information is fixed to $\gls{d_out} = 8$.

Each configuration is identified by a unique ID.
Configurations using \gls{dp} \condInfo are repeated for $\gls{eps_c} = 20$ and marked with an additional \texttt{\textcolor{RoyalBlue}{a}}, e.g., \texttt{5\textcolor{RoyalBlue}{a}} denotes the same setting as \texttt{5} but with a higher privacy budget for conditional information.
We refer to IDs \texttt{1}--\texttt{24} as the \textit{base configurations}, used to address the three main research questions.
Additionally, we include three ablation studies (IDs \texttt{25}--\texttt{39}) to isolate specific parameter effects.
The non-\gls{dl} baseline PrivTrace is ID \texttt{40}.
Due to high computational cost, ablation studies 2 and 3 are limited to the Porto dataset.

A preliminary hyperparameter search identified suitable learning rates, batch sizes, epochs, and clipping norms, aiming for consistency across configurations.
A full search with formal tuning was not feasible due to resource constraints (\refer \secref{eval_runtimes}).
A complete listing of all configurations is provided in \tabref{tab_all_cases} in the appendix.

\subsection{Results}\label{sec_results}

\begin{figure*}[ht]
    \centering
    \subfloat[
        Conditional models provide no formal guarantees (\ballnumber{A}).
        Unconditional models provide formal guarantees \wrt the test data (\ballnumber{C}).
        \label{fig_examples_nodp_porto}
    ]{
        \includegraphics[width=0.98\textwidth]{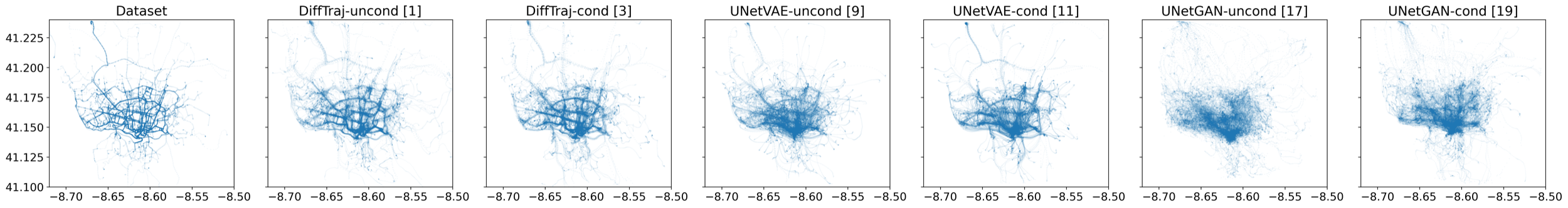}
    }

    \subfloat[
        All these models provide privacy guarantees \wrt the conditional information (\ballnumber{C}).
    \label{fig_examples_dpcond_ec10_porto}
    ]{
        \includegraphics[width=0.98\textwidth]{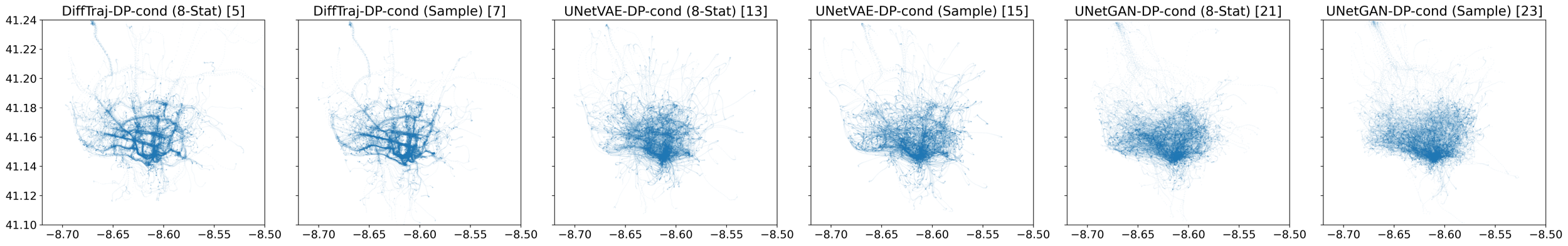}
    }

    \subfloat[
        All these models provide full formal guarantees (\ballnumber{D}).
    \label{fig_examples_dp_ec10_porto}
    ]{
        \includegraphics[width=0.98\textwidth]{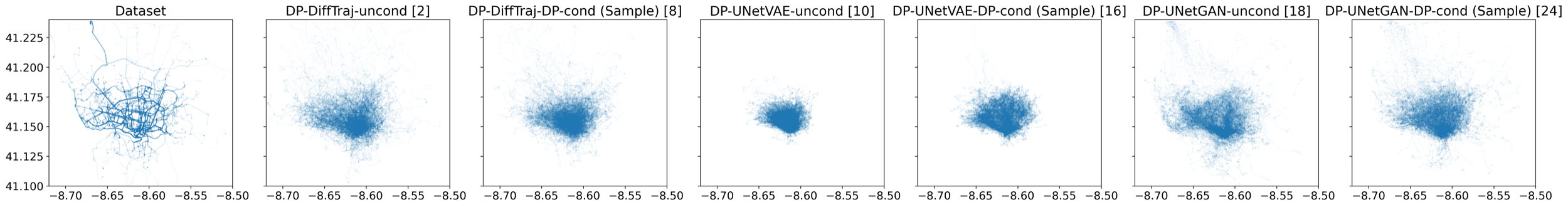}
    }

    \caption{
        Example synthetic trajectories for Porto after $\approx\num{100000}$ steps.
        The case IDs (\refer \tabref{tab_all_cases}) are reported in brackets.
    }
    \label{fig_combined_examples_porto}
    \vspace{-1em}
\end{figure*}

This section presents the evaluation results.
Figures~\ref{fig_swd_porto}--\ref{fig_geolife_swd_by_tm} show the \gls{swd} and \gls{jsd} for Porto and GeoLife, respectively.
Error bars indicate \SI{95}{\percent} confidence intervals, threat models (\refer \secref{sec_threat_model}) are shown in the legends, and PrivTrace is abbreviated as \texttt{PT}.
Confidence intervals are wide due to the limited number of runs (five) per configuration.
Additional runs were infeasible given the significant computational cost.

For our DP-\condInfo mechanism, we only show the \Sample format with Laplace noise to simplify the plots, as the \DTcond format yields similar results.
\figref{fig_combined_examples_porto} shows the point distribution for \num{1000} synthetic Porto trajectories.
\ifgeolife
  GeoLife examples are in the appendix (\figref{fig_combined_examples_geolife}).
\fi
We focus our discussion on \gls{swd} and \gls{jsd}, as these best reflect the observed distribution quality.
Due to the large number of configurations and metrics, the remaining results are reported in the appendix in \tabref{tab_all_results_porto} and \tabref{tab_all_results_gl}, respectively.

\subsubsection{Research Question 1}\label{eval_dpsgd}
\begin{figure}[b]
    \vspace{-1.5em}
    \includegraphics[width=\columnwidth]{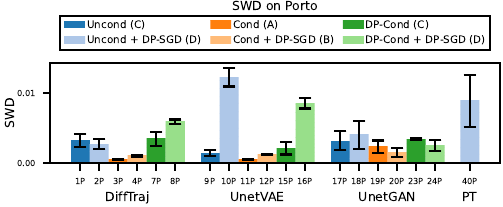}
    \vspace{-1.5em}
    \caption[SWD on Porto]{%
        The bars depict the \gls{swd} for the Porto Taxi dataset.
    }
    \label{fig_swd_porto}
\end{figure}

Our first research question asked: \rqDPSGD
To answer this question, we trained each base configuration without (threat models \ballnumber{A} and \ballnumber{C}) and with \gls{dp-sgd} (\ballnumber{B} and \ballnumber{D}) and compared the results.
Figures~\ref{fig_swd_porto}, \ref{fig_swd_geolife}, \ref{fig_jsd_porto}, and \ref{fig_jsd_geolife} show non-\gls{dp-sgd} cases in darker shades and \gls{dp-sgd} cases in lighter shades of the same base colour.
Comparing adjacent bars reveals that \gls{dp-sgd} degrades performance across models except for UNetGAN, which performs similarly or even slightly better with \gls{dp-sgd}.

On Porto, the \gls{swd} of the conditional DiffTraj  (ID \texttt{3} vs \texttt{4})  %
and UNetVAE (ID \texttt{11} vs \texttt{12})  %
approximately doubles despite using conditional information without \gls{dp}, and the unconditional \gls{vae} even increases its \gls{swd} by a factor of \num{8} (ID \texttt{9} vs \texttt{10})  %
(\refer \figref{fig_swd_porto}).
The unconditional DiffTraj's \gls{swd} (ID \texttt{1} vs \texttt{2}) appears to be an outlier as it's the only DiffTraj measurement that does not show an increase when using \gls{dp-sgd}.
Moreover, this case's \gls{jsd} increases 2.6-fold when using \gls{dp-sgd}. %
\ifgeolife
This aligns with the example figures (\figref{fig_combined_examples_porto} and \figref{fig_combined_examples_geolife}) which show a substantial utility loss for all DiffTraj and UNetVAE models when using \gls{dp-sgd}.
\fi

\begin{figure}[b]
    \vspace{-1.5em}
  \includegraphics[width=\columnwidth]{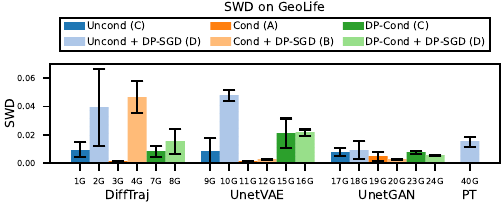}
  \vspace{-1.5em}
  \caption[SWD on GeoLife]{%
      The bars depict the \gls{swd} for the GeoLife dataset.
  }
  \label{fig_swd_geolife}
\end{figure}

Interestingly, the UNetGAN performs similar (unconditional, ID \texttt{17} vs \texttt{18}) or even better (conditional, ID \texttt{19} vs \texttt{20}) with \gls{dp-sgd} than without.
This result is partly due to the poor performance of the non-\gls{dp} model and the fact that \gls{dp}-UNetGAN can update the generator more often, as it does not require increasing $n_{critic}$ to $5$ for stability.
Nevertheless, the DP-UNetGAN outperforms the other models for threat models \ballnumber{B} and \ballnumber{D} regarding most point-level utility metrics on both datasets.
This interesting finding might indicate that the best model in the standard setting (DiffTraj) does not necessarily translate to the best model under \gls{dp} (UNetGAN).

\ifgeolife
  Visually, Figures~\ref{fig_combined_examples_porto} and~\ref{fig_combined_examples_geolife} show
\fi
that synthetic data generated with \gls{dp-sgd} retains the overall shape of the dataset but does not fully capture finer details, such as the road network.
The \gls{dp} version of unconditional DiffTraj (ID \texttt{2}) and both DP-UNetGAN variants (ID \texttt{18} and \texttt{24}) approximate the general distribution of the Porto dataset.
Hence, these models may still be suitable for applications that require only coarse-grained density information.

In contrast, the GeoLife results are largely unusable\xspace
\ifgeolife(\refer \figref{fig_examples_dp_ec10_geolife})\fi.
This can be attributed to the Porto dataset being over \SI{20}{\times} larger than GeoLife, as larger datasets generally yield better results under \gls{dp-sgd}~\cite{dpfyML}.
The bar plots for GeoLife (Figures~\ref{fig_swd_geolife} and \ref{fig_jsd_geolife}) also show much larger confidence intervals than those for Porto (Figures~\ref{fig_swd_porto} and~\ref{fig_jsd_porto}), indicating less stable training and greater variability between runs.

\update{
  Finally, while this paper focuses on deep learning models, we include PrivTrace as a \gls{sota} non-deep learning baseline for comparison.
  When comparing the results to PrivTrace (ID \texttt{40}, labelled \texttt{PT} in the figures), most \gls{dp-sgd} models outperform the baseline on Porto.
  On GeoLife, results are closer, but UNetGAN still surpasses PrivTrace.
  These findings show that deep learning models, despite being more recent, can achieve comparable utility to \gls{sota} non-\gls{dl} methods, underscoring their potential for trajectory generation.
  Still, future work is needed: in TM \ballnumber{D}, neither the baseline nor the models provide sufficient utility for most downstream applications, as evident from the example
  \ifgeolife
    figures~\ref{fig_combined_examples_porto} and~\ref{fig_combined_examples_geolife}.
  \fi
}

\begin{figure}[t]
  \includegraphics[width=\columnwidth]{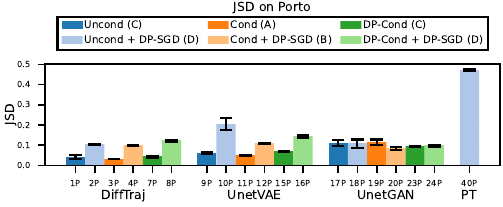}
  \vspace{-1.5em}
  \caption[JSD on Porto]{%
      The bars depict the \gls{jsd} for the Porto Taxi dataset.
  }
  \label{fig_jsd_porto}
  \vspace{-1.5em}
\end{figure}

\begin{figure}[b]
  \vspace{-1.5em}
  \includegraphics[width=\columnwidth]{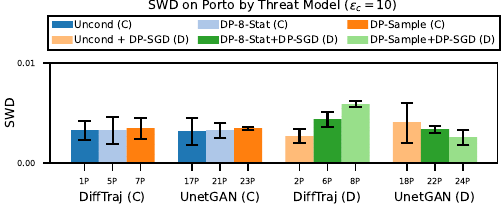}
  \vspace{-1.5em}
  \caption[JSD on Porto]{%
      The bars depict the \gls{swd} ordered by threat model.
  }
  \label{fig_porto_swd_by_tm}
\end{figure}

\subheading{Conclusion}
To answer \ref{rq_dpsgd}, our results show that \gls{dp-sgd} incurs substantial utility loss in trajectory generation, restricting most applications.
Nonetheless, performance remains comparable to non-\gls{dl} baselines such as PrivTrace, underscoring the potential of deep learning models.
Providing high utility using \gls{dp-sgd} (Threat models \ballnumber{B} and \ballnumber{D}) is not feasible with the current state-of-the-art models and privacy mechanisms and requires further research.

\subsubsection{Research Question 2}\label{eval_conditional}

Second, we asked: \rqDPCond
To answer this question, we proposed a method to provide formal guarantees \wrt the conditional information in \secref{sec_conditional_info}.
Here, we compare unconditional and conditional models using the \gls{dp} \condInfo mechanism for DiffTraj and UNetGAN on Porto (Figure~\ref{fig_porto_swd_by_tm}) and GeoLife (Figure~\ref{fig_geolife_swd_by_tm}), respectively.
The left groups represent threat model \ballnumber{C} (without \gls{dp-sgd}), and the right groups represent threat model \ballnumber{D} (with \gls{dp-sgd}).

Without \gls{dp-sgd} (TM \ballnumber{C}), all three configurations of both model types perform similarly, revealing no benefit of using the \gls{dp} \condInfo mechanism over an unconditional model.
However, when applying \gls{dp-sgd} (TM \ballnumber{D}), the results become more interesting.
For DiffTraj on Porto, the utility even degrades when using the \gls{dp} \condInfo mechanism with either input format.
However, for DiffTraj on GeoLife and for UNetGAN on both datasets, the \gls{dp} \condInfo mechanism improves utility compared to the unconditional model.
In these cases, the unconditional model performs worst, the \DTcond format yields intermediate results, and the \Sample format performs best.

For example, on Porto the UNetGAN using \gls{dp} \condInfo with \Sample format (ID \texttt{24}) represents the best model for threat model \ballnumber{D} regarding all point-level metrics, and most on GeoLife.
On GeoLife, the UNetGAN model exhibits noticeably improved point distributions when incorporating the \gls{dp} \condInfo mechanism, both without (IDs \texttt{23} vs \texttt{17}) and with \gls{dp-sgd} (IDs \texttt{24} vs \texttt{18}).

\begin{figure}[t]
  \includegraphics[width=\columnwidth]{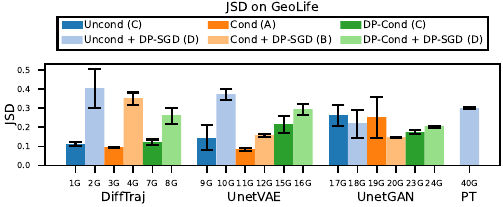}
  \vspace{-1.5em}
  \caption[JSD on GeoLife]{%
      The bars depict the \gls{jsd} for the GeoLife dataset.
  }
  \label{fig_jsd_geolife}
  \vspace{-1.5em}
\end{figure}

\begin{figure}[b]
  \vspace{-1.5em}
  \includegraphics[width=\columnwidth]{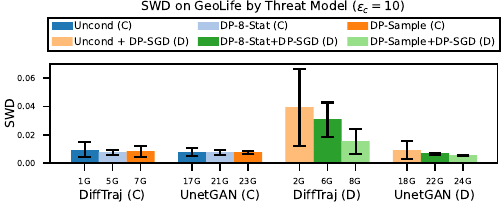}
  \vspace{-1.5em}
  \caption[JSD on Geolife]{%
      The bars depict the \gls{swd} ordered by threat model.
  }
  \label{fig_geolife_swd_by_tm}
\end{figure}

These results suggest that the \gls{dp} \condInfo mechanism primarily enhances model stability.
This is supported by two observations.
First, UNetGAN benefits the most, consistent with \glspl{gan} being notoriously unstable~\cite{wgan}.
Second, the effect is stronger on the smaller GeoLife dataset, as smaller datasets generally lead to more unstable training.

Moreover, we note that the \gls{dp} \condInfo mechanism becomes beneficial for TM \ballnumber{C} when relaxing the privacy parameter to $\varepsilon_c = 20$.
For example, without \gls{dp-sgd}, DiffTraj with \Sample DP \condInfo (ID \texttt{7a}) reduces the \gls{swd} by \SI{12}{\percent} on Porto  %
and \SI{30}{\percent} on GeoLife  %
compared to the unconditional model (ID \texttt{1}).
While $\varepsilon_c = 20$ may seem weak, we highlight that the \gls{dp} \condInfo mechanism is analysed \wrt the replace-one adjacency relation, which is considered approximately twice as strong as the common add-or-remove relation (\refer \secref{sec_privacy_combined}).
Additionally, the $\varepsilon_c$ values are computed under the worst-case assumption $m = n$, although we actually used $m = \num{3000} \ll n = D_{test}$ during evaluation.
Accordingly, the true privacy level is higher than reported, and for sufficiently large datasets, the mechanism can yield benefits even under stricter guarantees.

\subheading{Conclusion}
To answer \ref{rq_dpcond}, our results show that the \gls{dp} \condInfo mechanism generally improves utility when combined with \gls{dp-sgd} (TM \ballnumber{D}), especially for UNetGAN and on the smaller GeoLife dataset.
Without \gls{dp-sgd} (TM \ballnumber{C}), it offers no clear benefit under stricter privacy settings but can yield improvements when relaxing $\varepsilon_c$.
Overall, the mechanism enhances model stability and can improve utility, particularly for unstable models and smaller datasets.

\subsubsection{Research Question 3}\label{eval_architecture}

Finally, we asked: \rqArchitecture
Without formal guarantees (TM \ballnumber{A}), DiffTraj with conditional information (ID \texttt{3}) yields the best metrics and visual results for Porto, confirming that diffusion models represent the current \gls{sota} for trajectory generation.
However, the conditional UNetVAE (ID \texttt{11}) achieves similar results, even outperforming DiffTraj on several, mostly sequential, metrics.
On Porto, the synthetic data produced by UNetVAE appears more blurry than that of DiffTraj (\refer \figref{fig_examples_nodp_porto}).
On GeoLife, the conditional UNetVAE (ID \texttt{11}) outperforms DiffTraj (ID \texttt{3}) on all but two metrics.
This aligns with the general observation that diffusion models require larger datasets for optimal performance compared to \gls{vae} models (Porto contains \SI{20}{\times} more data than GeoLife).
The UNetGAN (IDs \texttt{17} and \texttt{19}) performs significantly worse than the other models for TM \ballnumber{A}, both visually and based on the \gls{swd} metric (\refer \figref{fig_swd_porto}).
The poorer performance may partly result from the reduced number of generator updates (\refer \secref{sec_architecture}).
However, increasing the updates by a factor of \SI{5} (\num{500000} total), to match the other models, did not yield comparable results.

Furthermore, the results highlight the critical role of conditional information in improving the utility of synthetic data across all models.
Visually, Figure~\ref{fig_examples_nodp_porto} shows a clear difference between conditional and unconditional models.
This is supported by the \gls{swd} values on Porto, which decrease by \SI{84}{\percent} for DiffTraj (ID \texttt{1} vs \texttt{3}),  %
\SI{64}{\percent} for UNetVAE (ID \texttt{9} vs \texttt{11}),  %
and \SI{27}{\percent} for UNetGAN (ID \texttt{17} vs \texttt{19})  %
when conditional information is used.
This trend is also evident in the corresponding bar plots \figref{fig_swd_porto} and \figref{fig_swd_geolife}.

We already discussed the different models in the context of \gls{dp-sgd} in \secref{eval_dpsgd}.
UNetGAN achieves the best performance under \gls{dp-sgd} on both datasets, followed by the diffusion model, while UNetVAE performs consistently poorly in all cases involving \gls{dp-sgd}.
These findings indicate that the model with the best non-\gls{dp} performance is not necessarily optimal when privacy guarantees are required.
The reasons for this behaviour remain unclear.
The \gls{gan} may benefit from \gls{dp-sgd} being applied only to the generator, which has roughly half the parameters of the other models (\refer \secref{sec_architecture}).
This results in less relative noise per gradient update, which may explain the improved performance.
\begin{figure}[t]
  \includegraphics[width=\columnwidth]{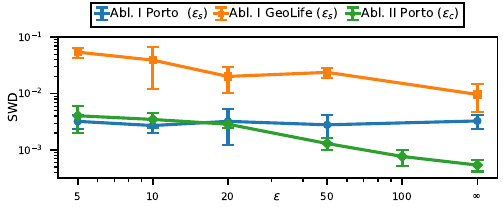}
  \vspace{-1.5em}
  \caption{
      Ablation studies I and II with \SI{95}{\percent} conf.\ Intervals over log-scale.
      $\infty$ indicates configurations without \gls{dp}.
  }
  \label{fig_ablation_eps}
  \vspace{-1em}
\end{figure}

\subheading{Conclusion}
Without \gls{dp-sgd}, the diffusion model DiffTraj achieves the best performance, followed closely by the UNetVAE.
However, when applying \gls{dp-sgd}, the UNetGAN outperforms both other models, indicating that the best model in the non-\gls{dp} setting is not necessarily optimal under \gls{dp}.

\subsubsection{Ablation Study I: Influence of \texorpdfstring{$\varepsilon_s$}{epsilon-s}}\label{eval_eps_s}
In our first ablation study (\figref{fig_ablation_eps}), we varied the privacy parameter $\varepsilon_s \in \left\{5, 10, 20, 50\right\}$ of \gls{dp-sgd} for the unconditional DiffTraj.
On GeoLife, we observe the expected trend: increasing $\varepsilon_s$ improves utility, as indicated by the decreasing \gls{swd} (orange curve in \figref{fig_ablation_eps}).
Even at $\varepsilon_s = 50$, the model performs significantly worse than without \gls{dp-sgd}, with an \gls{swd} that is \SI{2.5}{\times} higher than the non-\gls{dp} model (ID \texttt{1} vs \texttt{27}).  %
On Porto (blue curve), we do not observe a clear correlation between $\varepsilon_s$ and \gls{swd}.
However, other metrics such as \gls{jsd} confirm the expected trend, with the largest \gls{jsd} of \num{0.108} for $\varepsilon_s = 5$ (ID \texttt{25}) and the smallest \gls{jsd} of \num{0.0409} for the non-\gls{dp} model (ID \texttt{1}).
Hence, the \gls{swd} results on Porto appear to be an outlier.

\subsubsection{Ablation Study II: Influence of \texorpdfstring{$\varepsilon_c$}{epsilon-c}}\label{eval_eps_c}
Similarly, we varied the privacy parameter $\varepsilon_c \in \left\{5, 10, 20, 50, 100\right\}$ for the conditional DiffTraj model with \gls{dp} \condInfo of \Sample format (green curve in \figref{fig_ablation_eps}).
The figure shows a clear correlation between increasing $\varepsilon_c$ and decreasing \gls{swd}.
While $\varepsilon_c = 50$ (ID \texttt{30}) yields a \SI{140}{\percent} increase in \gls{swd} compared to the non-\gls{dp} model (ID \texttt{3}), $\varepsilon_c = 100$ (ID \texttt{31}) reduces the gap to \SI{42}{\percent}.  %

\begin{figure}[t]
  \includegraphics[width=\columnwidth]{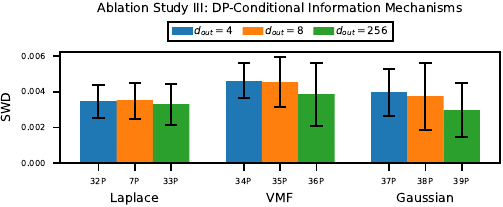}
  \vspace{-1.5em}
  \caption[Ablation Study III]{%
      Ablation Study III.
      Each group represents one of the three \gls{dp} \condInfo mechanisms, and the colours indicate the $d_{out}$ value.
  }
  \label{fig_ablation_iii}
  \vspace{-1em}
\end{figure}

\subsubsection{Ablation Study III: DP Mechanisms}\label{eval_mechanisms}
As the third ablation study, we compared the three \gls{dp} \condInfo mechanisms -- Laplace, Gaussian, and \gls{vmf}-- introduced in \secref{sec_conditional_info}, using $d_{out} \in \left\{4, 8, 256\right\}$ and $\varepsilon_c = 10$.
\figref{fig_ablation_iii} depicts the \gls{swd} results for the Porto dataset.
For fixed $\varepsilon_c$, all mechanisms perform similarly.
The only configuration that clearly outperforms the default Laplace mechanism with $d_{out} = 8$ (ID \texttt{7}) is the Gaussian mechanism with $d_{out} = 256$ (ID \texttt{39}).
However, this mechanism performs better only regarding point-based metrics, while degrading most sequential metrics.
For $d_{out} \in \left\{4, 8\right\}$, the Laplace mechanism remains the best choice, followed by the Gaussian mechanism.
Despite the dependency on the $l_1$-norm, which suggested worse performance for larger $d_{out}$, the Laplace mechanism performs consistently across all $d_{out}$ values.
In contrast, the Gaussian mechanism performs well only at $d_{out} = 256$.
The \gls{vmf} mechanism yields the poorest results among the three mechanisms.
While the mechanism improves with larger $d_{out}$, this comes at the cost of significantly increased runtime (\refer \secref{eval_runtimes}).
In conclusion, choosing between the Laplace mechanism at small $d_{out}$ and the Gaussian mechanism at large $d_{out}$ appears most promising.
The \gls{vmf} mechanism does not appear suitable for this task.

\subsubsection{Runtimes}\label{eval_runtimes}
All runtimes are based on the hardware specified in \secref{sec_hardware} and reported in \tabref{tab_all_cases} for Porto.
Geolife runtimes are slightly higher due to the increased trajectory length.
DiffTraj trains the longest, e.g., $\approx\SI{6}{\hour}$ for the conditional baseline, while the \gls{vae} and \gls{gan} require $\approx\SI{2}{\hour}$ and $\approx\SI{5}{\hour}$, respectively.
Matching the number of generator updates in UNetGAN to those of the other models increases the training time to $\approx\SI{25}{\hour}$, which justifies our choice to align total training steps instead. 
Unconditional models generally train significantly faster than their conditional counterparts, with minimal differences for VAE.

Using \gls{dp-sgd} increases runtimes by a factor of 5--10.
For instance, the conditional DiffTraj requires $\approx\SI{54}{\hour}$ (ID \texttt{4}), VAE $\approx\SI{29}{\hour}$, and GAN $\approx\SI{32}{\hour}$.
The smaller relative increase for UNetGAN is due to only the generator being trained with \gls{dp-sgd} such that the discriminator, thereby half of the total parameters, is updated without \gls{dp-sgd}.

The \gls{dp} \condInfo mechanism has negligible runtime impact.
Only the \gls{vmf} mechanism incurs significant overhead for large $d_{out}$, but this configuration is impractical due to poor performance (\refer \secref{eval_mechanisms}).
Notably, PrivTrace runtimes are significantly higher than reported in~\cite{PrivTrace2023} due to hard-coded parameters $K$ and $\kappa$ (\refer \secref{sec_setup}), which result in a larger amount of Markov states.

For all cases, evaluation took approximately \SI{9}{\min}, with most of the time spent on metric computation.
Generating \num{3000} samples took $<\SI{1}{\second}$ for the \gls{vae} and \gls{gan}, and $\approx\SI{20}{\second}$ for the diffusion model.
Generation is unaffected by \gls{dp-sgd}, as it applies only during training.

\section{Discussion}\label{sec_discussion}

Some authors suggest that the inherent obfuscation of the generative process alone suffices to protect privacy~\cite{DiffTraj23}.
However, this method does not provide formal privacy guarantees~\cite{pets24_paper}.
\gls{mia}~\cite{Shokri2017} can recover training data from a model, and an overfitted model may leak private information.
We acknowledge that recovering private information from synthetic data is challenging, and such obfuscation may suffice in low-risk scenarios.
However, using \condInfo during generation introduces a more significant risk.
For example, \citet{pets24_paper} showed that LSTM-TrajGAN can reproduce trajectories nearly identical to its inputs.

Guarantees \wrt generation inputs $D_{test}$ (TM \ballnumber{C}) are attainable today with moderate utility loss, e.g., using unconditional DiffTraj (ID \texttt{1}) or DiffTraj with \gls{dp} \condInfo (ID \texttt{5} or \texttt{7}).
In contrast, protecting both training and test data (TM \ballnumber{D}) leads to severe utility loss, rendering most models unsuitable for real-world applications, as clearly visible in Figure~\ref{fig_combined_examples_porto}.
Still, the utility for this threat model \ballnumber{D} remains comparable to \gls{sota} non-\gls{dl} baselines such as PrivTrace.
This highlights the potential of deep learning models, and future work on differentially private deep learning may render them viable for privacy-preserving data generation in high-risk settings in the future.

Meanwhile, a hybrid approach can be considered for models with generation-time inputs, such as conditional models.
Here, \gls{dp-sgd} is omitted, i.e., no guarantees are provided \wrt $D_{train}$, but the \condInfo ($D_{test}$) is protected via mechanisms such as the proposed \gls{dp} \condInfo mechanism or utilising unconditional models.
Although lacking formal guarantees towards the full dataset, this method can constitute a reasonable compromise in specific contexts, especially when combined with diffusion models, which show reduced susceptibility to \glspl{mia} (but remain vulnerable~\cite{Matsumoto2023}).
Mitigating direct links between generation inputs and outputs can reduce risks and direct leakage, such as those demonstrated in~\cite{pets24_paper}.

\subheading{Limitations}\label{sec_limitations}
Despite careful hyperparameter selection and preliminary testing, the high computational cost limited the number of configurations we could evaluate.
Therefore, we cannot guarantee that the optimal configuration or each model was found but are confident that the tested configurations offer a representative performance overview and reveal valid trends.
Although the final utility depends on factors such as available resources, dataset properties, and the use case, the results serve as a solid starting point for model selection.

\section{Conclusion}\label{sec_conclusion}

In this work, we set out to answer the question: \textit{\mainRQ}?
To this end, we evaluated the utility of synthetic trajectory data using three model types -- \gls{vae}, \gls{gan}, and diffusion models -- with and without \gls{dp-sgd} and conditional information.
The evaluation covered two real-world datasets and eleven metrics.
In total, we analysed over \num{90} configurations to examine the effect of various parameters on the utility-privacy trade-off.

The results revealed several key findings.
First, \gls{dp-sgd} significantly reduces the utility of generated trajectories.
While acceptable utility may still be achievable for specific use cases with sufficiently large datasets, full formal privacy guarantees remain restrictive for most practical applications.
Nonetheless, comparison to a \gls{sota} non-\gls{dl} baseline, PrivTrace, shows that deep learning models can achieve comparable or better utility under \gls{dp}, underscoring their potential for trajectory generation.
Second, we proposed a novel \gls{dp} mechanism to enable the use of conditional information with formal \gls{dp} guarantees \wrt the real data and formally proved its guarantees.
Although the \gls{dp} \condInfo mechanism does not consistently outperform unconditional models, it improves training stability especially combined with \gls{dp-sgd}, for unstable models, and on smaller datasets.
Third, the comparison of model types confirmed that, without formal privacy guarantees, diffusion models generally offer the best performance.
However, the evaluated \gls{vae} performed strongly on smaller datasets.
Under \gls{dp-sgd}, the \gls{gan} exhibited superior performance, demonstrating that the best model without privacy does not necessarily remain optimal when guarantees are applied.
All source code has been made available to support future research.

\ifanonymous
\else
  \ifacm
    \begin{acks}
  \else
    \section*{Acknowledgements}
  \fi
      The authors would like to thank UNSW, the Commonwealth of Australia,
      and the Cybersecurity Cooperative Research
      Centre Limited for their support of this work.
  \ifacm
    \end{acks}
  \fi
\fi

\bibliographystyle{IEEEtranNDOIandURLwithDate}
\bibliography{library}
\appendices

\section{All Results}

\ifgeolife
  In the following we provide detailed information on the used configurations (\tabref{tab_all_cases}), the results of the evaluation (\tabref{tab_all_results_porto} and \tabref{tab_all_results_gl}), and examples of synthetic trajectories for the second dataset GeoLife (\figref{fig_combined_examples_geolife}).
\fi

{
  \vbadness=10000 %
  \printglossary[type=symbols, title=Symbols]
  \printglossary[title=Abbreviations]
}

\begin{table*}
  \centering
  \caption{
    Overview of all evaluation configurations.
    The table lists the variable parameters of all test cases.
    Refer to \secref{sec_configurations} for an explanation of these values.
    The results for all cases are listed in \tabref{tab_all_results_porto} and \tabref{tab_all_results_gl}.
    The mean runtimes are reported for the Porto dataset.
    Runtimes on the GeoLife dataset are slightly higher due to the increased trajectory length.
    }
  \begin{adjustbox}{max width=\linewidth, max totalheight=0.905\textheight, keepaspectratio}
  \begin{tabular}{cccccccccc}
    \toprule
    \textbf{ID} & \textbf{Model} & \textbf{DP-SGD} & \textbf{$\varepsilon_s$} & \textbf{Cond. Info.} & \textbf{Cond. DP.} & \textbf{$\varepsilon_c$} & \textbf{$d_{out}$} & \textbf{\glsshort{TM}} & \textbf{Runtime} \\
    \midrule

    \multicolumn{10}{c}{ \textbf{ DiffTraj } } \\ \midrule
    1 & DiffTraj & \xmark & -- & None & -- & -- & -- & \ballnumber{C} & 03:04:58\unit{\hour} \\
    2 & DiffTraj & $\checkmark$ & 10.0 & None & -- & -- & -- & \ballnumber{D} & 27:59:13\unit{\hour} \\
    3 & DiffTraj & \xmark & -- & \DTcond & None & -- & -- & \ballnumber{A} & 05:55:18\unit{\hour} \\
    4 & DiffTraj & $\checkmark$ & 10.0 & \DTcond & None & -- & -- & \ballnumber{B} & 53:25:46\unit{\hour} \\
    5\textit{\textcolor{RoyalBlue}{(a)}} & DiffTraj & \xmark & -- & \DTcond & \gls{lap} & 10.0 / \textit{\textcolor{RoyalBlue}{20.0}} & 8 & \ballnumber{C} & 05:58:09\unit{\hour} \\
    6\textit{\textcolor{RoyalBlue}{(a)}} & DiffTraj & $\checkmark$ & 10.0 & \DTcond & \gls{lap} & 10.0 / \textit{\textcolor{RoyalBlue}{20.0}} & 8 & \ballnumber{D} & 53:26:53\unit{\hour} \\
    7\textit{\textcolor{RoyalBlue}{(a)}} & DiffTraj & \xmark & -- & \Sample & \gls{lap} & 10.0 / \textit{\textcolor{RoyalBlue}{20.0}} & 8 & \ballnumber{C} & 05:54:40\unit{\hour} \\
    8\textit{\textcolor{RoyalBlue}{(a)}} & DiffTraj & $\checkmark$ & 10.0 & \Sample & \gls{lap} & 10.0 / \textit{\textcolor{RoyalBlue}{20.0}} & 8 & \ballnumber{D} & 52:09:51\unit{\hour} \\

    \multicolumn{10}{c}{ \textbf{ UNetVAE } } \\ \midrule
    9 & UNetVAE & \xmark & -- & None & -- & -- & -- & \ballnumber{C} & 02:08:03\unit{\hour} \\
    10 & UNetVAE & $\checkmark$ & 10.0 & None & -- & -- & -- & \ballnumber{D} & 28:38:35\unit{\hour} \\
    11 & UNetVAE & \xmark & -- & \DTcond & None & -- & -- & \ballnumber{A} & 02:17:50\unit{\hour} \\
    12 & UNetVAE & $\checkmark$ & 10.0 & \DTcond & None & -- & -- & \ballnumber{B} & 33:01:06\unit{\hour} \\
    13\textit{\textcolor{RoyalBlue}{(a)}} & UNetVAE & \xmark & -- & \DTcond & \gls{lap} & 10.0 / \textit{\textcolor{RoyalBlue}{20.0}} & 8 & \ballnumber{C} & 02:20:26\unit{\hour} \\
    14\textit{\textcolor{RoyalBlue}{(a)}} & UNetVAE & $\checkmark$ & 10.0 & \DTcond & \gls{lap} & 10.0 / \textit{\textcolor{RoyalBlue}{20.0}} & 8 & \ballnumber{D} & 33:40:07\unit{\hour} \\
    15\textit{\textcolor{RoyalBlue}{(a)}} & UNetVAE & \xmark & -- & \Sample & \gls{lap} & 10.0 / \textit{\textcolor{RoyalBlue}{20.0}} & 8 & \ballnumber{C} & 02:18:22\unit{\hour} \\
    16\textit{\textcolor{RoyalBlue}{(a)}} & UNetVAE & $\checkmark$ & 10.0 & \Sample & \gls{lap} & 10.0 / \textit{\textcolor{RoyalBlue}{20.0}} & 8 & \ballnumber{D} & 31:49:42\unit{\hour} \\

    \multicolumn{10}{c}{ \textbf{ UNetGAN } } \\ \midrule
    17 & UNetGAN & \xmark & -- & None & -- & -- & -- & \ballnumber{C} & 04:41:02\unit{\hour} \\
    18 & UNetGAN & $\checkmark$ & 10.0 & None & -- & -- & -- & \ballnumber{D} & 28:40:34\unit{\hour} \\
    19 & UNetGAN & \xmark & -- & \DTcond & None & -- & -- & \ballnumber{A} & 05:03:54\unit{\hour} \\
    20 & UNetGAN & $\checkmark$ & 10.0 & \DTcond & None & -- & -- & \ballnumber{B} & 31:28:14\unit{\hour} \\
    21\textit{\textcolor{RoyalBlue}{(a)}} & UNetGAN & \xmark & -- & \DTcond & \gls{lap} & 10.0 / \textit{\textcolor{RoyalBlue}{20.0}} & 8 & \ballnumber{C} & 05:06:41\unit{\hour} \\
    22\textit{\textcolor{RoyalBlue}{(a)}} & UNetGAN & $\checkmark$ & 10.0 & \DTcond & \gls{lap} & 10.0 / \textit{\textcolor{RoyalBlue}{20.0}} & 8 & \ballnumber{D} & 31:43:34\unit{\hour} \\
    23\textit{\textcolor{RoyalBlue}{(a)}} & UNetGAN & \xmark & -- & \Sample & \gls{lap} & 10.0 / \textit{\textcolor{RoyalBlue}{20.0}} & 8 & \ballnumber{C} & 05:01:23\unit{\hour} \\
    24\textit{\textcolor{RoyalBlue}{(a)}} & UNetGAN & $\checkmark$ & 10.0 & \Sample & \gls{lap} & 10.0 / \textit{\textcolor{RoyalBlue}{20.0}} & 8 & \ballnumber{D} & 30:53:21\unit{\hour} \\

    \multicolumn{10}{c}{ \textbf{ Ablation Study I: Varying $\varepsilon_s$ } } \\ \midrule
    25 & DiffTraj & $\checkmark$ & 5.0 & None & -- & -- & -- & \ballnumber{D} & 27:54:51\unit{\hour} \\
    2 & DiffTraj & $\checkmark$ & 10.0 & None & -- & -- & -- & \ballnumber{D} & 27:59:13\unit{\hour} \\
    26 & DiffTraj & $\checkmark$ & 20.0 & None & -- & -- & -- & \ballnumber{D} & 28:11:25\unit{\hour} \\
    27 & DiffTraj & $\checkmark$ & 50.0 & None & -- & -- & -- & \ballnumber{D} & 28:19:53\unit{\hour} \\

    \multicolumn{10}{c}{ \textbf{ Ablation Study II: Varying $\varepsilon_c$ } } \\ \midrule
    28 & DiffTraj & \xmark & -- & \Sample & \gls{lap} & 5.0 & 8 & \ballnumber{C} & 05:53:38\unit{\hour} \\
    7 & DiffTraj & \xmark & -- & \Sample & \gls{lap} & 10.0 & 8 & \ballnumber{C} & 05:54:40\unit{\hour} \\
    29 & DiffTraj & \xmark & -- & \Sample & \gls{lap} & 20.0 & 8 & \ballnumber{C} & 05:54:31\unit{\hour} \\
    30 & DiffTraj & \xmark & -- & \Sample & \gls{lap} & 50.0 & 8 & \ballnumber{C} & 05:52:39\unit{\hour} \\
    31 & DiffTraj & \xmark & -- & \Sample & \gls{lap} & 100.0 & 8 & \ballnumber{C} & 05:54:18\unit{\hour} \\

    \multicolumn{10}{c}{ \textbf{ Ablation Study III: Varying \gls{dp}-Conditional Information Mechanisms } } \\ \midrule
    32 & DiffTraj & \xmark & -- & \Sample & \gls{lap} & 10.0 & 4 & \ballnumber{C} & 05:53:47\unit{\hour} \\
    7 & DiffTraj & \xmark & -- & \Sample & \gls{lap} & 10.0 & 8 & \ballnumber{C} & 05:54:40\unit{\hour} \\
    33 & DiffTraj & \xmark & -- & \Sample & \gls{lap} & 10.0 & 256 & \ballnumber{C} & 06:01:03\unit{\hour} \\
    34 & DiffTraj & \xmark & -- & \Sample & \gls{V} & 10.0 & 4 & \ballnumber{C} & 06:08:50\unit{\hour} \\
    35 & DiffTraj & \xmark & -- & \Sample & \gls{V} & 10.0 & 8 & \ballnumber{C} & 06:07:07\unit{\hour} \\
    36 & DiffTraj & \xmark & -- & \Sample & \gls{V} & 10.0 & 256 & \ballnumber{C} & 13:36:42\unit{\hour} \\
    37 & DiffTraj & \xmark & -- & \Sample & \gls{gau} & 10.0 & 4 & \ballnumber{C} & 05:54:44\unit{\hour} \\
    38 & DiffTraj & \xmark & -- & \Sample & \gls{gau} & 10.0 & 8 & \ballnumber{C} & 05:55:41\unit{\hour} \\
    39 & DiffTraj & \xmark & -- & \Sample & \gls{gau} & 10.0 & 256 & \ballnumber{C} & 05:57:36\unit{\hour} \\

    \multicolumn{10}{c}{ \textbf{ Non-Deep Learning Baseline: PrivTrace } } \\ \midrule
    40 & PrivTrace & -- & 10.0 & -- & -- & -- & -- & \ballnumber{D} & 93:36:17\unit{\hour} \\
    \bottomrule
\end{tabular}

  \end{adjustbox}
\label{tab_all_cases}
\end{table*}

\begin{table*}
  \centering
  \caption{
      This table lists the metric results for all Porto cases.
      Refer to \tabref{tab_all_cases} for the configuration of each case.
      The best baseline case (1--24) for each threat model is highlighted in a specific colour.
      The \textcolor{RoyalBlue}{(a)} cases use $\varepsilon_c = 20$ and are not considered for the computation of these best values.
      However, if one of the \textcolor{RoyalBlue}{(a)} cases beats the respective best case the value is bolded.
      }
  \begin{adjustbox}{max width=\linewidth, max totalheight=0.905\textheight, keepaspectratio}
  \begin{tabular}{ccccccccccccc}
    \toprule
    \textbf{ID} & \textbf{\glsshort{TM}} & \textbf{JSD ($\downarrow$)} & \textbf{SWD ($\downarrow$)} & \textbf{HD (P) ($\downarrow$)} & \textbf{Range ($\downarrow$)} & \textbf{Hotspot ($\uparrow$)} & \textbf{HD (T) ($\downarrow$)} & \textbf{Haversine ($\downarrow$)} & \textbf{DTW ($\downarrow$)} & \textbf{TTD ($\downarrow$)} & \textbf{TD ($\downarrow$)} & \textbf{TE ($\downarrow$)} \\
    \midrule

    \multicolumn{13}{c}{ \textbf{ DiffTraj } } \\ \midrule
    1P & \cellcolor{printblue}\ballnumber{C} & \textcolor{printblue}{\textbf{0.0409}} & 0.00328 & \textcolor{printblue}{\textbf{2124}} & 0.111 & 0.896 & 1198 & 665 & 5.42e+04 & 0.0171 & 0.0143 & 0.216 \\
    2P & \cellcolor{printgreen}\ballnumber{D} & 0.105 & 0.00272 & 3133 & 0.151 & 0.833 & \textcolor{printgreen}{\textbf{1255}} & \textcolor{printgreen}{\textbf{713}} & \textcolor{printgreen}{\textbf{6.37e+04}} & 0.156 & 0.0349 & 0.298 \\
    3P & \cellcolor{printred}\ballnumber{A} & \textcolor{printred}{\textbf{0.0313}} & \textcolor{printred}{\textbf{0.000541}} & 1617 & \textcolor{printred}{\textbf{0.0575}} & \textcolor{printred}{\textbf{0.958}} & 867 & 741 & 5.17e+04 & 0.047 & \textcolor{printred}{\textbf{0.00461}} & 0.154 \\
    4P & \cellcolor{printorange}\ballnumber{B} & 0.097 & \textcolor{printorange}{\textbf{0.00108}} & 2320 & 0.135 & 0.865 & 1252 & 977 & 8.1e+04 & 0.223 & 0.0153 & 0.257 \\
    5P & \cellcolor{printblue}\ballnumber{C} & 0.041 & 0.00327 & 2684 & \textcolor{printblue}{\textbf{0.0992}} & 0.914 & 2745 & 2119 & 2.04e+05 & 0.0603 & 0.0237 & 0.211 \\
    6P & \cellcolor{printgreen}\ballnumber{D} & 0.116 & 0.00439 & 4101 & 0.157 & 0.825 & 2542 & 1922 & 1.85e+05 & 0.0183 & 0.0276 & 0.294 \\
    7P & \cellcolor{printblue}\ballnumber{C} & 0.0419 & 0.00349 & 2475 & 0.101 & 0.912 & 2773 & 1924 & 1.84e+05 & \textcolor{printblue}{\textbf{0.0147}} & 0.0135 & \textcolor{printblue}{\textbf{0.208}} \\
    8P & \cellcolor{printgreen}\ballnumber{D} & 0.122 & 0.00591 & 4528 & 0.162 & 0.832 & 2325 & 1644 & 1.58e+05 & \textcolor{printgreen}{\textbf{0.0168}} & 0.0285 & 0.286 \\
    \midrule

    \rowcolor{gray!10}\textit{\textcolor{RoyalBlue}{5aP}} & \cellcolor{printblue}\ballnumber{C} & \textbf{0.0382} & 0.00222 & 2401 & \textbf{0.0869} & 0.919 & 2167 & 1675 & 1.55e+05 & 0.0437 & 0.0212 & 0.211 \\
    \rowcolor{gray!10}\textit{\textcolor{RoyalBlue}{6aP}} & \cellcolor{printgreen}\ballnumber{D} & 0.111 & 0.00352 & 4174 & 0.155 & 0.831 & 2165 & 1617 & 1.52e+05 & 0.0255 & 0.0276 & 0.286 \\
    \rowcolor{gray!10}\textit{\textcolor{RoyalBlue}{7aP}} & \cellcolor{printblue}\ballnumber{C} & 0.041 & 0.00287 & 2700 & \textbf{0.0957} & 0.921 & 2052 & 1349 & 1.24e+05 & 0.0179 & 0.0144 & \textbf{0.204} \\
    \rowcolor{gray!10}\textit{\textcolor{RoyalBlue}{8aP}} & \cellcolor{printgreen}\ballnumber{D} & 0.106 & 0.00391 & 4511 & 0.15 & \textbf{0.842} & 1678 & 1132 & 1.02e+05 & 0.0236 & 0.0207 & 0.269 \\

    \multicolumn{13}{c}{ \textbf{ UNetVAE } } \\ \midrule
    9P & \cellcolor{printblue}\ballnumber{C} & 0.0611 & \textcolor{printblue}{\textbf{0.00153}} & 2395 & 0.1 & \textcolor{printblue}{\textbf{0.921}} & \textcolor{printblue}{\textbf{1029}} & \textcolor{printblue}{\textbf{552}} & \textcolor{printblue}{\textbf{4.66e+04}} & 0.0199 & 0.0157 & 0.231 \\
    10P & \cellcolor{printgreen}\ballnumber{D} & 0.205 & 0.0123 & 7458 & 0.221 & 0.825 & 2169 & 1326 & 1.25e+05 & 0.0362 & 0.0312 & 0.326 \\
    11P & \cellcolor{printred}\ballnumber{A} & 0.0498 & 0.000548 & \textcolor{printred}{\textbf{1319}} & 0.066 & 0.922 & \textcolor{printred}{\textbf{782}} & \textcolor{printred}{\textbf{656}} & \textcolor{printred}{\textbf{4.56e+04}} & \textcolor{printred}{\textbf{0.0261}} & 0.0107 & \textcolor{printred}{\textbf{0.149}} \\
    12P & \cellcolor{printorange}\ballnumber{B} & 0.108 & 0.00129 & 2048 & 0.106 & 0.842 & \textcolor{printorange}{\textbf{897}} & \textcolor{printorange}{\textbf{682}} & \textcolor{printorange}{\textbf{5.29e+04}} & 0.0408 & 0.0166 & 0.201 \\
    13P & \cellcolor{printblue}\ballnumber{C} & 0.0686 & 0.00287 & 2923 & 0.118 & 0.892 & 3584 & 2853 & 2.81e+05 & 0.0196 & 0.0128 & 0.234 \\
    14P & \cellcolor{printgreen}\ballnumber{D} & 0.15 & 0.00812 & 6922 & 0.2 & 0.805 & 2625 & 2052 & 2.01e+05 & 0.0322 & 0.028 & 0.297 \\
    15P & \cellcolor{printblue}\ballnumber{C} & 0.0685 & 0.00215 & 2802 & 0.109 & 0.894 & 3906 & 2912 & 2.87e+05 & 0.0209 & \textcolor{printblue}{\textbf{0.0125}} & 0.236 \\
    16P & \cellcolor{printgreen}\ballnumber{D} & 0.142 & 0.00858 & 7022 & 0.205 & 0.827 & 2499 & 1760 & 1.73e+05 & 0.0579 & 0.0284 & 0.315 \\
    \midrule

    \rowcolor{gray!10}\textit{\textcolor{RoyalBlue}{13aP}} & \cellcolor{printblue}\ballnumber{C} & 0.0669 & 0.00291 & 2319 & 0.113 & 0.906 & 3052 & 2350 & 2.28e+05 & 0.0197 & 0.0151 & 0.23 \\
    \rowcolor{gray!10}\textit{\textcolor{RoyalBlue}{14aP}} & \cellcolor{printgreen}\ballnumber{D} & 0.12 & 0.00541 & 6689 & 0.18 & 0.817 & 2039 & 1688 & 1.63e+05 & 0.033 & 0.0287 & 0.297 \\
    \rowcolor{gray!10}\textit{\textcolor{RoyalBlue}{15aP}} & \cellcolor{printblue}\ballnumber{C} & 0.069 & 0.0028 & 2439 & 0.115 & 0.9 & 2877 & 2035 & 1.99e+05 & 0.019 & 0.0186 & 0.236 \\
    \rowcolor{gray!10}\textit{\textcolor{RoyalBlue}{16aP}} & \cellcolor{printgreen}\ballnumber{D} & 0.107 & 0.00573 & 6387 & 0.176 & \textbf{0.851} & 2025 & 1440 & 1.4e+05 & 0.0299 & 0.0233 & 0.295 \\

    \multicolumn{13}{c}{ \textbf{ UNetGAN } } \\ \midrule
    17P & \cellcolor{printblue}\ballnumber{C} & 0.109 & 0.0032 & 2634 & 0.157 & 0.802 & 1350 & 796 & 6.97e+04 & 0.0174 & 0.0194 & 0.327 \\
    18P & \cellcolor{printgreen}\ballnumber{D} & 0.107 & 0.00404 & 2285 & 0.161 & 0.822 & 1270 & 774 & 6.69e+04 & 0.0224 & 0.0179 & 0.282 \\
    19P & \cellcolor{printred}\ballnumber{A} & 0.114 & 0.00234 & 2060 & 0.141 & 0.817 & 1110 & 892 & 6.98e+04 & 0.0286 & 0.0141 & 0.276 \\
    20P & \cellcolor{printorange}\ballnumber{B} & \textcolor{printorange}{\textbf{0.0837}} & 0.00154 & \textcolor{printorange}{\textbf{1725}} & \textcolor{printorange}{\textbf{0.101}} & \textcolor{printorange}{\textbf{0.867}} & 935 & 783 & 5.87e+04 & \textcolor{printorange}{\textbf{0.0343}} & \textcolor{printorange}{\textbf{0.0107}} & \textcolor{printorange}{\textbf{0.199}} \\
    21P & \cellcolor{printblue}\ballnumber{C} & 0.0887 & 0.00324 & 3831 & 0.132 & 0.863 & 2564 & 1873 & 1.78e+05 & 0.0243 & 0.0144 & 0.254 \\
    22P & \cellcolor{printgreen}\ballnumber{D} & 0.102 & 0.00336 & 3483 & 0.147 & 0.83 & 2862 & 2121 & 2.04e+05 & 0.0249 & \textcolor{printgreen}{\textbf{0.0154}} & 0.276 \\
    23P & \cellcolor{printblue}\ballnumber{C} & 0.0948 & 0.00345 & 4242 & 0.148 & 0.855 & 2557 & 1788 & 1.71e+05 & 0.0172 & 0.0143 & 0.252 \\
    24P & \cellcolor{printgreen}\ballnumber{D} & \textcolor{printgreen}{\textbf{0.096}} & \textcolor{printgreen}{\textbf{0.00255}} & \textcolor{printgreen}{\textbf{2005}} & \textcolor{printgreen}{\textbf{0.134}} & \textcolor{printgreen}{\textbf{0.84}} & 2701 & 1874 & 1.79e+05 & 0.0218 & 0.0185 & \textcolor{printgreen}{\textbf{0.268}} \\
    \midrule

    \rowcolor{gray!10}\textit{\textcolor{RoyalBlue}{21aP}} & \cellcolor{printblue}\ballnumber{C} & 0.0868 & 0.00296 & 4060 & 0.13 & 0.87 & 2147 & 1557 & 1.44e+05 & 0.0248 & 0.0163 & 0.263 \\
    \rowcolor{gray!10}\textit{\textcolor{RoyalBlue}{22aP}} & \cellcolor{printgreen}\ballnumber{D} & \textbf{0.0866} & \textbf{0.00179} & 2386 & \textbf{0.122} & \textbf{0.848} & 2120 & 1556 & 1.43e+05 & 0.0295 & 0.0163 & \textbf{0.253} \\
    \rowcolor{gray!10}\textit{\textcolor{RoyalBlue}{23aP}} & \cellcolor{printblue}\ballnumber{C} & 0.088 & \textbf{0.00147} & 2879 & 0.123 & 0.855 & 2025 & 1364 & 1.26e+05 & 0.0209 & 0.0176 & 0.255 \\
    \rowcolor{gray!10}\textit{\textcolor{RoyalBlue}{24aP}} & \cellcolor{printgreen}\ballnumber{D} & \textbf{0.09} & \textbf{0.00138} & 2164 & \textbf{0.127} & \textbf{0.853} & 2097 & 1410 & 1.3e+05 & 0.0221 & 0.0182 & 0.272 \\

    \multicolumn{13}{c}{ \textbf{ Ablation Study I: Varying $\varepsilon_s$ } } \\ \midrule
    25P & \cellcolor{printgreen}\ballnumber{D} & 0.108 & 0.00323 & 3279 & 0.16 & 0.832 & 1290 & 742 & 6.63e+04 & 0.173 & 0.0348 & 0.303 \\
    2P & \cellcolor{printgreen}\ballnumber{D} & 0.105 & 0.00272 & 3133 & 0.151 & 0.833 & 1255 & 713 & 6.37e+04 & 0.156 & 0.0349 & 0.298 \\
    26P & \cellcolor{printgreen}\ballnumber{D} & 0.104 & 0.00323 & 2975 & 0.157 & 0.828 & 1280 & 736 & 6.54e+04 & 0.114 & 0.028 & 0.295 \\
    27P & \cellcolor{printgreen}\ballnumber{D} & 0.1 & 0.0028 & 2785 & 0.146 & 0.83 & 1238 & 702 & 6.26e+04 & 0.0579 & 0.0256 & 0.286 \\

    \multicolumn{13}{c}{ \textbf{ Ablation Study II: Varying $\varepsilon_c$ } } \\ \midrule
    28P & \cellcolor{printblue}\ballnumber{C} & 0.0475 & 0.00406 & 2588 & 0.104 & 0.871 & 3509 & 2545 & 2.49e+05 & 0.0133 & 0.00927 & 0.217 \\
    7P & \cellcolor{printblue}\ballnumber{C} & 0.0419 & 0.00349 & 2475 & 0.101 & 0.912 & 2773 & 1924 & 1.84e+05 & 0.0147 & 0.0135 & 0.208 \\
    29P & \cellcolor{printblue}\ballnumber{C} & 0.041 & 0.00287 & 2700 & 0.0957 & 0.921 & 2052 & 1349 & 1.24e+05 & 0.0179 & 0.0144 & 0.204 \\
    30P & \cellcolor{printblue}\ballnumber{C} & 0.0357 & 0.0013 & 2544 & 0.0766 & 0.943 & 1182 & 725 & 5.87e+04 & 0.0176 & 0.0141 & 0.188 \\
    31P & \cellcolor{printblue}\ballnumber{C} & 0.0313 & 0.00077 & 2117 & 0.0649 & 0.964 & 769 & 456 & 3.33e+04 & 0.0204 & 0.0133 & 0.169 \\

    \multicolumn{13}{c}{ \textbf{ Ablation Study III: Varying \gls{dp}-Conditional Information Mechanisms } } \\ \midrule
    32P & \cellcolor{printblue}\ballnumber{C} & 0.0444 & 0.00347 & 2789 & 0.103 & 0.894 & 2788 & 1951 & 1.87e+05 & 0.149 & 0.0342 & 0.209 \\
    7P & \cellcolor{printblue}\ballnumber{C} & 0.0419 & 0.00349 & 2475 & 0.101 & 0.912 & 2773 & 1924 & 1.84e+05 & 0.0147 & 0.0135 & 0.208 \\
    33P & \cellcolor{printblue}\ballnumber{C} & 0.0437 & 0.00328 & 2814 & 0.102 & 0.89 & 2882 & 2026 & 1.95e+05 & 0.029 & 0.0143 & 0.208 \\
    34P & \cellcolor{printblue}\ballnumber{C} & 0.0505 & 0.00462 & 2492 & 0.12 & 0.888 & 2925 & 2085 & 2.02e+05 & 0.152 & 0.0342 & 0.216 \\
    35P & \cellcolor{printblue}\ballnumber{C} & 0.0515 & 0.00454 & 2638 & 0.116 & 0.89 & 3312 & 2374 & 2.32e+05 & 0.0154 & 0.0201 & 0.216 \\
    36P & \cellcolor{printblue}\ballnumber{C} & 0.0454 & 0.00384 & 2071 & 0.108 & 0.892 & 4073 & 3085 & 3.06e+05 & 0.0296 & 0.0129 & 0.212 \\
    37P & \cellcolor{printblue}\ballnumber{C} & 0.0456 & 0.00397 & 3869 & 0.109 & 0.889 & 3855 & 2884 & 2.85e+05 & 0.277 & 0.0492 & 0.21 \\
    38P & \cellcolor{printblue}\ballnumber{C} & 0.0464 & 0.00373 & 2337 & 0.108 & 0.873 & 3973 & 2956 & 2.92e+05 & 0.13 & 0.0149 & 0.209 \\
    39P & \cellcolor{printblue}\ballnumber{C} & 0.0406 & 0.00298 & 2362 & 0.0999 & 0.913 & 4100 & 3084 & 3.05e+05 & 0.0454 & 0.0126 & 0.206 \\

    \multicolumn{13}{c}{ \textbf{ Non-Deep Learning Baseline: PrivTrace } } \\ \midrule
    40P & \cellcolor{printgreen}\ballnumber{D} & 0.472 & 0.00892 & 5015 & 0.247 & 0.671 & 2562 & 1451 & 1.42e+05 & 0.25 & 0.538 & 0.534 \\
    \bottomrule
\end{tabular}

  \end{adjustbox}
  \label{tab_all_results_porto}
\end{table*}

\begin{table*}
  \centering
  \caption{
      This table lists the metric results for all GeoLife cases.
      Refer to \tabref{tab_all_cases} for the configuration of each case.
      The best baseline case (1--24) for each threat model is highlighted in a specific colour.
      The \textcolor{RoyalBlue}{(a)} cases use $\varepsilon_c = 20$ and are not considered for the computation of these best values.
      However, if one of the \textcolor{RoyalBlue}{(a)} cases beats the respective best case the value is bolded.
      }
  \begin{adjustbox}{max width=\linewidth, max totalheight=\textheight, keepaspectratio}
  \begin{tabular}{ccccccccccccc}
    \toprule
    \textbf{ID} & \textbf{\glsshort{TM}} & \textbf{JSD ($\downarrow$)} & \textbf{SWD ($\downarrow$)} & \textbf{HD (P) ($\downarrow$)} & \textbf{Range ($\downarrow$)} & \textbf{Hotspot ($\uparrow$)} & \textbf{HD (T) ($\downarrow$)} & \textbf{Haversine ($\downarrow$)} & \textbf{DTW ($\downarrow$)} & \textbf{TTD ($\downarrow$)} & \textbf{TD ($\downarrow$)} & \textbf{TE ($\downarrow$)} \\
    \midrule

    \multicolumn{13}{c}{ \textbf{ DiffTraj } } \\ \midrule
    1G & \cellcolor{printblue}\ballnumber{C} & \textcolor{printblue}{\textbf{0.113}} & 0.00967 & 5755 & \textcolor{printblue}{\textbf{0.115}} & \textcolor{printblue}{\textbf{0.826}} & 1655 & 1192 & 2.31e+05 & 0.104 & \textcolor{printblue}{\textbf{0.00788}} & \textcolor{printblue}{\textbf{0.139}} \\
    2G & \cellcolor{printgreen}\ballnumber{D} & 0.404 & 0.0392 & 8167 & 0.435 & 0.262 & 5797 & 5159 & 1.03e+06 & \textcolor{printgreen}{\textbf{0.103}} & 0.205 & 0.408 \\
    3G & \cellcolor{printred}\ballnumber{A} & 0.0924 & 0.00116 & 3376 & 0.0877 & \textcolor{printred}{\textbf{0.907}} & 871 & 777 & 1.41e+05 & 0.089 & 0.00826 & 0.0962 \\
    4G & \cellcolor{printorange}\ballnumber{B} & 0.35 & 0.0466 & 6416 & 0.345 & 0.484 & 1e+04 & 8887 & 1.78e+06 & 0.0327 & 0.202 & 0.414 \\
    5G & \cellcolor{printblue}\ballnumber{C} & 0.13 & \textcolor{printblue}{\textbf{0.00743}} & 6179 & 0.131 & 0.794 & 4775 & 4151 & 8.28e+05 & 0.172 & 0.0104 & 0.157 \\
    6G & \cellcolor{printgreen}\ballnumber{D} & 0.34 & 0.0306 & 6729 & 0.338 & 0.401 & 9170 & 8188 & 1.64e+06 & 0.299 & 0.394 & 0.434 \\
    7G & \cellcolor{printblue}\ballnumber{C} & 0.122 & 0.00807 & 6062 & 0.119 & 0.804 & 4224 & 3565 & 7.12e+05 & 0.164 & 0.0105 & 0.142 \\
    8G & \cellcolor{printgreen}\ballnumber{D} & 0.26 & 0.0152 & \textcolor{printgreen}{\textbf{5349}} & 0.24 & 0.565 & 5319 & 4348 & 8.68e+05 & 0.243 & 0.321 & 0.336 \\
    \midrule

    \rowcolor{gray!10}\textit{\textcolor{RoyalBlue}{5aG}} & \cellcolor{printblue}\ballnumber{C} & 0.143 & 0.00867 & 6412 & 0.145 & 0.789 & 3741 & 3250 & 6.47e+05 & 0.154 & 0.0109 & 0.163 \\
    \rowcolor{gray!10}\textit{\textcolor{RoyalBlue}{6aG}} & \cellcolor{printgreen}\ballnumber{D} & 0.321 & 0.026 & 6786 & 0.328 & 0.437 & 7711 & 6705 & 1.34e+06 & 0.267 & 0.394 & 0.424 \\
    \rowcolor{gray!10}\textit{\textcolor{RoyalBlue}{7aG}} & \cellcolor{printblue}\ballnumber{C} & 0.126 & \textbf{0.00677} & 6467 & \textbf{0.113} & \textbf{0.84} & 2889 & 2312 & 4.6e+05 & 0.0892 & 0.00883 & 0.139 \\
    \rowcolor{gray!10}\textit{\textcolor{RoyalBlue}{8aG}} & \cellcolor{printgreen}\ballnumber{D} & 0.237 & 0.0104 & 5379 & 0.198 & 0.594 & 3643 & 2719 & 5.41e+05 & 0.212 & 0.334 & 0.33 \\

    \multicolumn{13}{c}{ \textbf{ UNetVAE } } \\ \midrule
    9G & \cellcolor{printblue}\ballnumber{C} & 0.146 & 0.00872 & 8643 & 0.129 & 0.779 & \textcolor{printblue}{\textbf{1601}} & \textcolor{printblue}{\textbf{1147}} & \textcolor{printblue}{\textbf{2.23e+05}} & \textcolor{printblue}{\textbf{0.0374}} & 0.0203 & 0.176 \\
    10G & \cellcolor{printgreen}\ballnumber{D} & 0.372 & 0.0475 & 2.38e+04 & 0.181 & 0.564 & 5322 & 4738 & 9.47e+05 & 0.44 & 0.278 & 0.363 \\
    11G & \cellcolor{printred}\ballnumber{A} & \textcolor{printred}{\textbf{0.0833}} & \textcolor{printred}{\textbf{0.00105}} & 3802 & \textcolor{printred}{\textbf{0.0837}} & 0.903 & \textcolor{printred}{\textbf{802}} & \textcolor{printred}{\textbf{695}} & \textcolor{printred}{\textbf{1.25e+05}} & \textcolor{printred}{\textbf{0.00998}} & \textcolor{printred}{\textbf{0.00706}} & \textcolor{printred}{\textbf{0.091}} \\
    12G & \cellcolor{printorange}\ballnumber{B} & 0.156 & \textcolor{printorange}{\textbf{0.00258}} & \textcolor{printorange}{\textbf{5253}} & 0.121 & \textcolor{printorange}{\textbf{0.841}} & \textcolor{printorange}{\textbf{1451}} & 1068 & 2.1e+05 & 0.0382 & \textcolor{printorange}{\textbf{0.07}} & \textcolor{printorange}{\textbf{0.168}} \\
    13G & \cellcolor{printblue}\ballnumber{C} & 0.183 & 0.0153 & 1.15e+04 & 0.163 & 0.661 & 5482 & 4881 & 9.76e+05 & 0.0523 & 0.0207 & 0.197 \\
    14G & \cellcolor{printgreen}\ballnumber{D} & 0.266 & 0.0327 & 2.04e+04 & 0.217 & 0.575 & 5323 & 4636 & 9.26e+05 & 0.352 & 0.232 & 0.291 \\
    15G & \cellcolor{printblue}\ballnumber{C} & 0.215 & 0.0211 & 1.35e+04 & 0.159 & 0.672 & 5053 & 4495 & 8.99e+05 & 0.12 & 0.0353 & 0.218 \\
    16G & \cellcolor{printgreen}\ballnumber{D} & 0.294 & 0.0216 & 2.07e+04 & 0.241 & 0.537 & 4546 & 3907 & 7.8e+05 & 0.335 & 0.287 & 0.301 \\
    \midrule

    \rowcolor{gray!10}\textit{\textcolor{RoyalBlue}{13aG}} & \cellcolor{printblue}\ballnumber{C} & 0.17 & 0.0124 & 1.21e+04 & 0.158 & 0.728 & 3128 & 2611 & 5.22e+05 & 0.0844 & 0.0321 & 0.173 \\
    \rowcolor{gray!10}\textit{\textcolor{RoyalBlue}{14aG}} & \cellcolor{printgreen}\ballnumber{D} & 0.225 & 0.0207 & 1.79e+04 & 0.231 & 0.574 & 3934 & 3287 & 6.55e+05 & 0.346 & 0.245 & 0.267 \\
    \rowcolor{gray!10}\textit{\textcolor{RoyalBlue}{15aG}} & \cellcolor{printblue}\ballnumber{C} & 0.176 & 0.0128 & 1.23e+04 & 0.152 & 0.708 & 2848 & 2321 & 4.64e+05 & 0.129 & 0.0377 & 0.182 \\
    \rowcolor{gray!10}\textit{\textcolor{RoyalBlue}{16aG}} & \cellcolor{printgreen}\ballnumber{D} & 0.222 & 0.0139 & 1.59e+04 & 0.193 & 0.668 & 3379 & 2691 & 5.36e+05 & 0.334 & 0.286 & 0.249 \\

    \multicolumn{13}{c}{ \textbf{ UNetGAN } } \\ \midrule
    17G & \cellcolor{printblue}\ballnumber{C} & 0.264 & 0.0079 & \textcolor{printblue}{\textbf{5169}} & 0.193 & 0.66 & 2075 & 1311 & 2.55e+05 & 0.0801 & 0.1 & 0.277 \\
    18G & \cellcolor{printgreen}\ballnumber{D} & 0.22 & 0.00932 & 8201 & \textcolor{printgreen}{\textbf{0.168}} & \textcolor{printgreen}{\textbf{0.702}} & \textcolor{printgreen}{\textbf{1937}} & \textcolor{printgreen}{\textbf{1297}} & \textcolor{printgreen}{\textbf{2.56e+05}} & 0.194 & 0.229 & \textcolor{printgreen}{\textbf{0.193}} \\
    19G & \cellcolor{printred}\ballnumber{A} & 0.251 & 0.00498 & \textcolor{printred}{\textbf{3282}} & 0.168 & 0.744 & 1408 & 1108 & 2.07e+05 & 0.0588 & 0.0715 & 0.234 \\
    20G & \cellcolor{printorange}\ballnumber{B} & \textcolor{printorange}{\textbf{0.148}} & 0.00293 & 7158 & \textcolor{printorange}{\textbf{0.118}} & 0.839 & 1459 & \textcolor{printorange}{\textbf{1056}} & \textcolor{printorange}{\textbf{2.03e+05}} & \textcolor{printorange}{\textbf{0.0239}} & 0.126 & 0.169 \\
    21G & \cellcolor{printblue}\ballnumber{C} & 0.172 & 0.00776 & 8399 & 0.152 & 0.752 & 4154 & 3492 & 6.97e+05 & 0.0842 & 0.0942 & 0.215 \\
    22G & \cellcolor{printgreen}\ballnumber{D} & \textcolor{printgreen}{\textbf{0.193}} & 0.00669 & 1.08e+04 & 0.176 & 0.688 & 4149 & 3546 & 7.08e+05 & 0.171 & \textcolor{printgreen}{\textbf{0.194}} & 0.207 \\
    23G & \cellcolor{printblue}\ballnumber{C} & 0.175 & 0.00749 & 8038 & 0.162 & 0.774 & 4129 & 3352 & 6.69e+05 & 0.118 & 0.0848 & 0.222 \\
    24G & \cellcolor{printgreen}\ballnumber{D} & 0.203 & \textcolor{printgreen}{\textbf{0.00547}} & 8987 & 0.176 & 0.695 & 4063 & 3480 & 6.95e+05 & 0.173 & 0.212 & 0.203 \\
    \midrule

    \rowcolor{gray!10}\textit{\textcolor{RoyalBlue}{21aG}} & \cellcolor{printblue}\ballnumber{C} & 0.181 & \textbf{0.00682} & 9029 & 0.163 & 0.74 & 3212 & 2628 & 5.24e+05 & 0.135 & 0.113 & 0.214 \\
    \rowcolor{gray!10}\textit{\textcolor{RoyalBlue}{22aG}} & \cellcolor{printgreen}\ballnumber{D} & \textbf{0.176} & \textbf{0.00533} & 1.05e+04 & \textbf{0.159} & \textbf{0.746} & 2932 & 2341 & 4.67e+05 & 0.178 & 0.226 & 0.2 \\
    \rowcolor{gray!10}\textit{\textcolor{RoyalBlue}{23aG}} & \cellcolor{printblue}\ballnumber{C} & 0.165 & \textbf{0.00366} & \textbf{5113} & 0.14 & 0.818 & 3047 & 2398 & 4.76e+05 & 0.0798 & 0.0603 & 0.193 \\
    \rowcolor{gray!10}\textit{\textcolor{RoyalBlue}{24aG}} & \cellcolor{printgreen}\ballnumber{D} & 0.195 & \textbf{0.00509} & 9157 & 0.168 & \textbf{0.742} & 2860 & 2279 & 4.54e+05 & 0.193 & 0.226 & 0.201 \\

    \multicolumn{13}{c}{ \textbf{ Ablation Study I: Varying $\varepsilon_s$ } } \\ \midrule
    25G & \cellcolor{printgreen}\ballnumber{D} & 0.432 & 0.0543 & 5823 & 0.46 & 0.08 & 7875 & 6774 & 1.35e+06 & 0.121 & 0.278 & 0.482 \\
    2G & \cellcolor{printgreen}\ballnumber{D} & 0.404 & 0.0392 & 8167 & 0.435 & 0.262 & 5797 & 5159 & 1.03e+06 & 0.103 & 0.205 & 0.408 \\
    26G & \cellcolor{printgreen}\ballnumber{D} & 0.289 & 0.0202 & 7129 & 0.244 & 0.571 & 3135 & 2577 & 5.14e+05 & 0.0543 & 0.187 & 0.278 \\
    27G & \cellcolor{printgreen}\ballnumber{D} & 0.263 & 0.0238 & 6766 & 0.196 & 0.635 & 3244 & 2757 & 5.5e+05 & 0.135 & 0.0831 & 0.239 \\

    \multicolumn{13}{c}{ \textbf{ Non-Deep Learning Baseline: PrivTrace } } \\ \midrule
    40G & \cellcolor{printgreen}\ballnumber{D} & 0.3 & 0.015 & 1.11e+04 & 0.167 & 0.793 & 2211 & 1671 & 3.32e+05 & 0.155 & 0.323 & 0.169 \\
    \bottomrule
\end{tabular}

  \end{adjustbox}
  \label{tab_all_results_gl}
\end{table*}

\ifgeolife
  \begin{figure*}
      \centering
      \subfloat[
          Conditional models provide no formal guarantees (\ballnumber{A}).
          Unconditional models provide formal guarantees \wrt the test data (\ballnumber{C}).
      \label{fig_examples_nodp_geolife}
      ]{
          \includegraphics[width=0.98\textwidth]{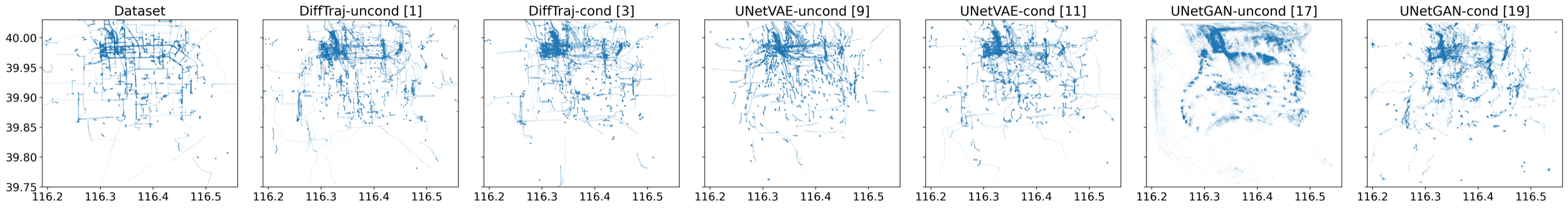}
      }

      \subfloat[
          All these models provide privacy guarantees \wrt the conditional information (\ballnumber{C}).
      \label{fig_examples_dpcond_ec10_geolife}
      ]{
          \includegraphics[width=0.98\textwidth]{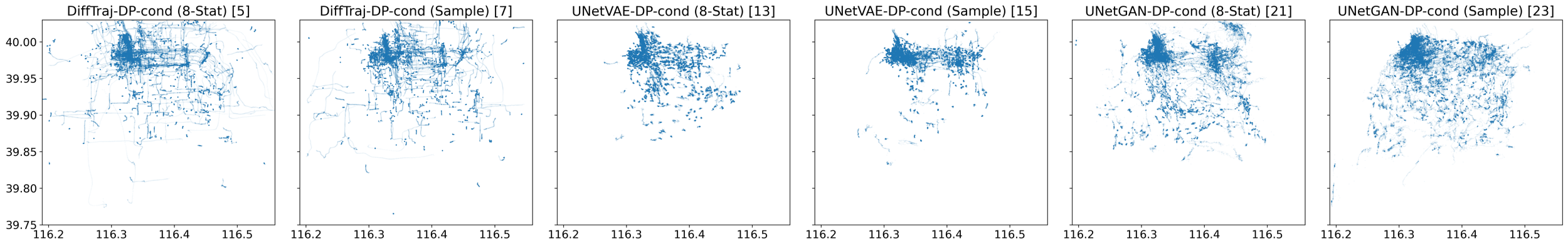}
      }

      \subfloat[
          All these models provide full formal guarantees (\ballnumber{D}).
      \label{fig_examples_dp_ec10_geolife}
      ]{
          \includegraphics[width=0.98\textwidth]{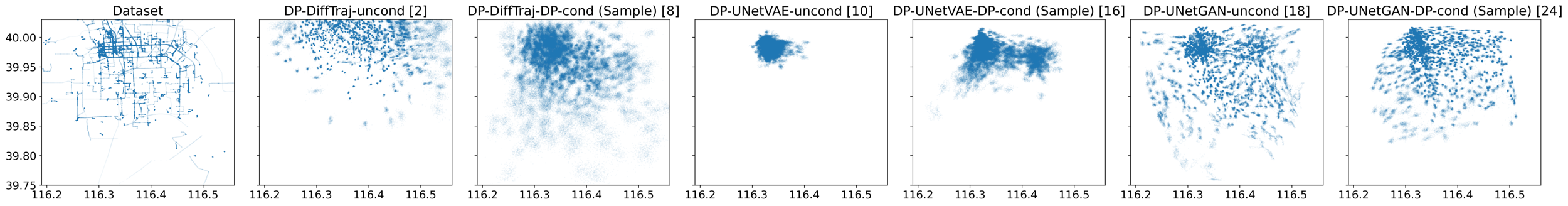}
      }

      \caption{
          Example synthetic trajectories for GeoLife after $\approx\num{100000}$ steps.
          The case IDs (\refer \tabref{tab_all_cases}) are reported in brackets.
      }
      \label{fig_combined_examples_geolife}
      \vspace{-1.5em}
  \end{figure*}
\fi

\end{document}